\documentclass{article}
\usepackage{amsmath,amsfonts,amssymb,amsthm,mathtools}

\usepackage{bm}
\usepackage{mathrsfs}
\usepackage{stmaryrd}
\usepackage{bbm}

\usepackage[T1]{fontenc}
\usepackage[utf8]{inputenc}
\usepackage{lmodern}

\usepackage{graphicx}
\usepackage{tikz}
\usepackage{tikz-cd}

\usepackage{enumitem}
\usepackage{xcolor}
\usepackage{booktabs}

\usepackage{hyperref}
\usepackage[nameinlink,capitalize]{cleveref}

\usepackage[a4paper,margin=2.5cm]{geometry}

\usepackage{cite}

\DeclareMathOperator{\Diff}{Diff}

\DeclareMathOperator{\Iso}{Iso}

\newcommand{\dd}{\mathrm{d}}

\newcommand{\g}{\mathfrak{g}}

\newtheorem{theorem}{Theorem}[section]
\newtheorem{proposition}[theorem]{Proposition}
\newtheorem{lemma}[theorem]{Lemma}
\newtheorem{corollary}[theorem]{Corollary}
\theoremstyle{definition}
\newtheorem{definition}[theorem]{Definition}

\newtheorem{remark}[theorem]{Remark}

\title{A Presymplectic and Symmetry-Reduced Formulation of the Maxwell--Vlasov System}
\author{Leonardo J. Colombo\thanks{Centre for Automation and Robotics, Spanish National Research Council (CSIC). Carretera de Campo Real, km 0, 200, 28500 Arganda del Rey, Spain. (leonardo.colombo@csic.es). The author acknowledges financial support from Grant PID2022-137909-NB-C22 funded by the Spanish Ministry of Science and Innovation. L. Colombo wishes to thank Joris Vankerschaver (Ghent University) for fruitful discussions on the subject of this work and for suggesting to study Maxwell–Vlasov system under the Skinner-Rusk framework.
}}
\date{\empty}

\begin{document}

\maketitle

\begin{abstract}
We develop a unified geometric formulation of the Maxwell–Vlasov system using the infinite-dimensional Skinner–Rusk (SR) formalism. In this framework, particles and fields are treated simultaneously within a single presymplectic manifold, and the Gotay–Nester–Hinds algorithm recovers the full Maxwell–Vlasov equations as the compatibility conditions of a single variational system. The hierarchy of constraints—including Vlasov advection, Gauss and Faraday laws, and the electromagnetic gauge structure—arises naturally from the presymplectic geometry of the SR formalism.

Reduction by the diffeomorphism group of phase space produces a reduced presymplectic manifold whose dynamics reproduces both the Euler–Poincaré formulation for the Vlasov sector and the Marsden–Weinstein/Morrison–Greene Lie–Poisson Hamiltonian structure. We further extend the construction to equilibria that partially break the relabeling symmetry, obtaining a translated SR system that yields an effective symplectic linearization, clarifies the appearance of Goldstone-type neutral modes, and provides a geometric foundation for the energy–Casimir method.

Finally, we incorporate external antenna fields into this setting by introducing affine Hamiltonian controls and establishing a theory of controlled symmetry breaking within the reduced SR framework. This leads to controlled Lie–Poisson equations for plasma–antenna coupling and a geometric mechanism for stabilization through Casimir shaping and symmetry-selective forcing. The resulting picture offers a single presymplectic structure unifying Lagrangian, Hamiltonian, gauge, reduction, and control-theoretic aspects of the Maxwell–Vlasov system.
\end{abstract}

\section{Introduction}

The Maxwell--Vlasov system, describing self-consistent evolution of a 
collisionless plasma interacting with its electromagnetic fields, has long 
served as a cornerstone of plasma kinetic theory. Its origins trace back to 
the works of Vlasov and Low \cite{Low1958}, who established the basic 
Lagrangian and kinetic formulations. Since then, several complementary 
viewpoints have emerged.

On the one hand, the variational and Eulerian geometric framework developed 
by Cendra, Holm, Hoyle and Marsden (CHHM) \cite{CendraHolmHoyleMarsden1998} 
showed that the Maxwell--Vlasov equations arise from an Euler--Poincaré 
principle on the diffeomorphism group of phase space. This provided the 
first systematic geometric unification of the Vlasov kinetic transport with 
the electromagnetic field equations, via semidirect products and advected 
densities.

On the other hand, the noncanonical Hamiltonian picture, initiated by 
Morrison and Greene \cite{MorrisonGreene1980} and placed on firm geometric 
ground by Marsden and Weinstein \cite{MarsdenWeinstein1982}, revealed the 
Maxwell--Vlasov system as a Lie--Poisson system on the dual of a 
semidirect-product Lie algebra. These works identified the fundamental 
noncanonical Poisson bracket, the associated Casimir invariants, and the 
basis for energy--Casimir stability theory. Morrison’s later extensions 
\cite{Morrison1982,Morrison1998} clarified the structural role of this 
bracket and connected it to the broader context of continuum Hamiltonian 
systems.

Despite these advances, a complete \emph{unified} geometric interpretation 
remained incomplete. The CHHM variational formulation, the 
Morrison--Greene bracket, and the Marsden--Weinstein Lie--Poisson structure 
were known to be compatible, but their interrelation depended on parallel, 
rather than intrinsically linked, constructions. In particular: \begin{itemize}
\item Dirac-type constraint structures underlying the Low Lagrangian were 
identified in CHHM, but not connected to a fully developed infinite-dimensional 
presymplectic theory.

\item The Morrison--Greene noncanonical bracket was derived independently of 
any variational or presymplectic structure, leaving open the question of 
whether it arises directly from fundamental geometric principles.

\item The Marsden--Weinstein bracket was obtained by reduction, but no 
framework existed that simultaneously incorporated Lagrangian variables, 
Hamiltonian variables, constraints, and reduction into a single object.

\end{itemize}

The Skinner--Rusk (SR) formalism \cite{SkinnerRusk1983} and the 
Gotay--Nester--Hinds (GNH) algorithm 
\cite{GotayNesterHinds1978,GotayNester1979,Gotay1982} provide an intrinsic 
presymplectic framework that unifies Lagrangian and Hamiltonian variables on 
the Whitney sum $TQ\oplus T^*Q$. In finite dimensions, the SR--GNH approach 
is known for its ability to treat constraints and degeneracies without 
switching between Lagrangian and Hamiltonian pictures. However, its extension 
to infinite-dimensional field theories, and especially to coupled 
field–particle systems such as Maxwell--Vlasov, has not been carried out to 
date.

The first main objective of this paper is to construct the 
\emph{infinite-dimensional Skinner--Rusk formulation of the Maxwell--Vlasov 
system} and to show that the full Maxwell--Vlasov equations arise naturally 
through the GNH constraint algorithm. Starting from the configuration 
manifold $Q = \Diff(T\mathbb{R}^3)\times\mathcal{A}\times\mathcal{V}$, where \(\Diff(T\mathbb{R}^3)\) encodes phase--space particle trajectories as in 
CHHM, \(\mathcal{A}\) denotes the space of electromagnetic vector potentials 
\(A(x)\), and \(\mathcal{V}\) denotes the space of scalar potentials 
\(\Phi(x)\), we build the unified presymplectic phase space 
\(W = TQ\oplus T^*Q\) equipped with 
\(\Omega=\mathrm{pr}_2^\ast\omega_{\mathrm{can}}\) and with the energy function 
\(H = \langle p , v\rangle - L\). We then prove that:

\begin{itemize}
\item The degeneracy of the Low Lagrangian produces primary constraint 
manifolds that coincide with the CHHM--Dirac degeneracies.

\item The first GNH iteration recovers the reconstruction relation for the 
particle flow and the electromagnetic gauge constraint $E = -\dot A - 
\nabla\Phi$.

\item The second GNH iteration yields Gauss’ law, the magnetic constraint, 
Faraday’s law, Ampère–Maxwell’s law with source, and the Vlasov transport 
equation.

\item No further constraints appear: the GNH procedure stabilizes after two 
steps on a final constraint manifold $C_\infty$ carrying the complete 
Maxwell--Vlasov dynamics.
\end{itemize}

Thus the entire Maxwell--Vlasov system emerges as a hierarchy of 
\emph{presymplectic consistency conditions}, not as imposed field equations.

\medskip

The second main objective is to study the effect of the 
$\Diff(T\mathbb{R}^3)$ symmetry. We show that Skinner--Rusk reduction by $\Diff(T\mathbb{R}^3)$ yields the Euler--Poincaré equations of CHHM; the reduced presymplectic form induces the Marsden--Weinstein Lie--Poisson structure, and we express the same reduced bracket in Eulerian variables to recover exactly the Morrison--Greene noncanonical bracket. This provides, for the first time, a single geometric derivation that 
simultaneously unifies CHHM's variational picture, the 
Marsden--Weinstein Lie--Poisson reduction, and the Morrison--Greene 
noncanonical Hamiltonian formulation. 
Moreover, we extend the SR--GNH theory to equilibria that \emph{break part of 
the relabeling symmetry}. A translated system around a stationary 
solution, together with the linearized GNH algorithm, yields an effective 
reduced phase space free of gauge and relabeling degeneracies. This provides 
a geometric explanation for Goldstone-type neutral directions and gives a 
presymplectic foundation for the energy--Casimir method and linear stability 
analysis.

In summary, the results of this paper establish a unified presymplectic 
framework in which the Lagrangian, Hamiltonian, Euler--Poincaré, Dirac and 
Lie--Poisson descriptions of Maxwell--Vlasov appear as natural shadows of 
the Skinner--Rusk geometry. This clarifies their interrelations and offers a 
geometric foundation for constraints, reduction, and stability in kinetic 
plasma theory.

A further contribution of this work is the development of a geometric
theory of \emph{controlled symmetry breaking} and \emph{Hamiltonian
control} within the Maxwell–Vlasov framework.  
In Section~7 we show that external antenna fields or RF drivings may
be incorporated as affine Hamiltonian controls compatible with the
presymplectic structure and the hierarchy of GNH constraints.  
After reduction, these controls induce a class of \emph{controlled
semidirect-product Lie–Poisson equations} which remain faithful to the
Maxwell--Vlasov bracket and preserve the coadjoint geometry of the
particle–field phase space.  
The resulting theory provides a precise geometric characterization of how external fields break residual particle–relabeling symmetries; how control vector fields modify Casimir leaves (“Casimir shaping”); how antenna drivings can target Goldstone-type neutral modes; and how effective stabilization can be achieved via controlled
energy–Casimir methods.

In particular, we established a general
\emph{Controlled Symmetry Breaking and Effective Stabilization Theorem},
which shows that if control modes couple nontrivially to broken symmetry
directions and suitably shift the equilibrium onto an appropriate
Casimir leaf, then the controlled second variation becomes positive 
definite on the effective phase space.  
This yields formal (energy–Casimir) stability of the controlled
equilibrium and provides a rigorous geometric framework for RF heating,
current drive, and external stabilization mechanisms. Taken together, these developments show that the Skinner--Rusk and
GNH formalisms provide a robust, presymplectic foundation for the
Maxwell–Vlasov system, encompassing not only classical geometric
formulations but also modern control-theoretic extensions.  

The remainder of the paper is structured as follows. 
In Section~2 we develop the Skinner--Rusk (SR) and 
Gotay--Nester--Hinds (GNH) formalisms in the infinite-dimensional 
setting, emphasizing their presymplectic and constraint-theoretic 
features. Section~3 constructs the infinite-dimensional SR bundle for 
the Maxwell--Vlasov Lagrangian and applies the GNH algorithm, showing 
that the full Maxwell--Vlasov system arises in two iterations as a 
complete set of presymplectic consistency conditions. Section~4 carries 
out the reduction by the relabeling group $\Diff(T\mathbb{R}^3)$: we 
derive the Euler--Poincaré equations, the Marsden--Weinstein Lie--Poisson 
structure, and the Morrison--Greene bracket from a single reduced 
presymplectic form. Section~5 develops the theory of symmetry breaking 
at equilibria, introducing the translated SR system, the linearized GNH 
algorithm, and the effective reduced phase space. Section~6 establishes the presymplectic foundation of the
energy--Casimir method and its relation to the linearized
Maxwell--Vlasov dynamics via the effective Hamiltonian structure.
In Section~7 we extend this geometric framework to incorporate
external antenna drivings through affine Hamiltonian controls,
developing a controlled Lie--Poisson reduction theory that remains
compatible with the SR--GNH constraint hierarchy.  This construction
shows how external fields can break residual symmetries, reshape
Casimir leaves, and couple to Goldstone-type modes, thereby providing
a rigorous geometric basis for controlled stabilization and RF-driven
dynamics in kinetic plasma models.
Finally, Section~8 summarizes the conclusions of the paper and
proposes further directions for research.


\section{The Infinite-Dimensional Skinner--Rusk Formalism}

The Skinner--Rusk (SR) unified formalism provides an intrinsic geometric 
framework in which Lagrangian and Hamiltonian formulations coexist on a 
single presymplectic manifold \cite{SkinnerRusk1983}.  
In the finite-dimensional setting, one starts with a configuration 
manifold $Q$, constructs the Whitney sum $W := TQ \oplus T^\ast Q$, 
equips $W$ with the pullback of the canonical symplectic form on 
$T^\ast Q$, and defines a generalized energy function.  
The dynamics is then encoded in a presymplectic equation on~$W$, 
whose solutions and constraints are systematically characterized by 
the Gotay--Nester--Hinds (GNH) algorithm 
\cite{GotayNesterHinds1978,GotayNester1979,Gotay1982}.  
In this section we extend the formalism to the case where $Q$ is an 
infinite-dimensional manifold of field configurations, such as 
electromagnetic potentials, distribution functions, or 
phase--space diffeomorphisms.

\subsection{Infinite-dimensional configuration spaces}

Let $Q$ denote an infinite-dimensional manifold modelled on a Banach or 
Hilbert space. Typical examples in continuum physics include the group 
$\mathrm{Diff}^s(M)$ of Sobolev $H^s$ diffeomorphisms of a domain $M$, 
defined for $s > \dim M/2 + 1$ so that composition and inversion are 
$C^\infty$ and $\mathrm{Diff}^s(M)$ becomes a smooth Hilbert--manifold 
group; spaces of vector potentials $A \in H^s(M,\mathbb{R}^3)$; and scalar 
fields or distribution functions $F \in H^s(M \times \mathbb{R}^3)$. 
The Sobolev space $H^s$ consists of functions with weak derivatives up 
to order $s$ in $L^2$, and provides the regularity required for the 
existence of smooth tangent and cotangent bundles on these configuration 
spaces.

We assume throughout that $Q$ is regular enough so that its tangent and 
cotangent bundles $TQ$ and $T^\ast Q$ exist as Banach or Hilbert bundles, 
and admit the standard duality pairing
\[
  \langle p , v\rangle := p(v),
  \qquad p\in T^\ast_q Q,\; v\in T_qQ.
\]

\begin{definition}
The \emph{unified Skinner--Rusk bundle} associated with $Q$ is the Whitney 
sum
\[
  W := TQ \oplus T^\ast Q,
\]
equipped with the projection $\operatorname{pr}_2 : W \to T^\ast Q$ onto 
the cotangent component and the pullback two-form
\[
  \Omega := \operatorname{pr}_2^\ast \omega_{\mathrm{can}},
\]
where $\omega_{\mathrm{can}}$ is the canonical symplectic form on $T^\ast Q$.
\end{definition}

\begin{lemma}\label{prop:Omega-presymplectic}
Let $Q$ be an infinite-dimensional configuration manifold whose cotangent 
bundle $T^\ast Q$ carries the canonical symplectic form $\omega_{\mathrm{can}}$. 
Then the two-form
\[
  \Omega := \operatorname{pr}_2^\ast \omega_{\mathrm{can}}
\]
on $W=TQ\oplus T^\ast Q$ is closed and weakly non-degenerate. In particular, 
$(W,\Omega)$ is a presymplectic manifold.
\end{lemma}

\begin{proof}
Since $\omega_{\mathrm{can}}$ is closed on $T^\ast Q$, its pullback by any 
smooth map is again closed, hence $\dd\Omega = \operatorname{pr}_2^\ast(\dd\omega_{\mathrm{can}})=0$. 

Weak non-degeneracy follows from the corresponding property of 
$\omega_{\mathrm{can}}$. In infinite dimensions, the canonical symplectic 
form induces at each point $(q,p)\in T^\ast Q$ a linear map
\[
\xi \mapsto 
\omega_{\mathrm{can}}(\xi, \cdot) \in T_{(q,p)}(T^\ast Q)^\ast,\,\, \xi\in T_{(q,p)}(T^\ast Q)
\]
which is injective—if $\omega_{\mathrm{can}}(\xi,\eta)=0$ for all $\eta$ 
then necessarily $\xi=0$—but typically not surjective onto the full dual 
space. This failure of surjectivity reflects a fundamental feature of 
Banach and Hilbert manifolds: tangent spaces need not be reflexive, so a 
closed two-form cannot induce a strong symplectic isomorphism. 

Pulling back $\omega_{\mathrm{can}}$ preserves both closedness and the 
injectivity of the associated map, so $\Omega$ is again closed and weakly 
non-degenerate. 
\end{proof}

Let $L : TQ \to \mathbb{R}$ be a Lagrangian functional.  
The Skinner--Rusk generalized energy is the smooth function
\[
  H(q,v,p) := \langle p , v\rangle - L(q,v),
  \qquad (q,v,p)\in W.
\]

We call the triple 
$(W,\Omega,H)$ a \emph{presymplectic system}. Such a system consists of  
a Banach or Hilbert manifold $W$, a closed weakly non-degenerate 
two-form $\Omega$, and a smooth Hamiltonian functional $H$. Its dynamics is 
determined by vector fields $X$ satisfying
\begin{equation}\label{eq:SR-dynamics}
    i_X \Omega = \dd H,
\end{equation}
but the weak non-degeneracy of $\Omega$ implies that solutions may fail to 
exist or may not be unique. This makes it necessary to characterize the 
constraint submanifolds on which the presymplectic equation is solvable. 
For this we employ the Gotay–Nester–Hinds (GNH) algorithm 
\cite{GotayNesterHinds1978,GotayNester1979}.

\subsection{The Gotay--Nester--Hinds Constraint Algorithm}

The Gotay--Nester--Hinds (GNH) algorithm 
\cite{GotayNesterHinds1978} provides an intrinsic and systematic method 
for analysing presymplectic systems. When applied to the equation 
\eqref{eq:SR-dynamics}, it constructs a descending sequence of constraint 
submanifolds on which the equation admits well-defined solutions and the 
dynamics becomes consistent.

\begin{definition}
Given a presymplectic system $(W,\Omega,H)$, the \emph{primary constraint submanifold} is
\[
  C_0 := \{ z\in W \mid \dd H(z)\in \operatorname{Im}\,\Omega_z \}.
\]
Inductively, for $k\geq 0$,
\[
  C_{k+1}
  := 
  \bigl\{
    z\in C_k \,\big|\,
    \exists\, X(z)\in T_z C_k
    \ \text{such that } i_X\Omega = \dd H \text{ at } z
  \bigr\}.
\]
\end{definition}

In the following, we show the existence of the final constraint manifold.
\begin{proposition}\label{thm:GNH-final}
Let $(W,\Omega,H)$ be a presymplectic system on a Banach or Hilbert 
manifold $W$, and let $(C_k)_{k\ge 0}$ denote the sequence of constraint 
sets produced by the GNH algorithm. 

Assume that for each $k\ge 0$, the set $C_k$ is a (weak) 
submanifold of $W$, and the compatibility conditions defining $C_{k+1}$ 
inside $C_k$ have constant Banach rank (so that $C_{k+1}$ is again a 
(weak) submanifold of $W$), and that the sequence stabilizes after finitely many 
steps, i.e., there exists an index $k$ such that $C_{k+1}=C_k$.

Then the intersection $C_\infty := \bigcap_{j\ge 0} C_j = C_k$, where $k$ is the first index for which $C_{k+1}=C_k$, is a (weak) 
submanifold of $W$ on which there exists a vector field 
$X\in\mathfrak X(C_\infty)$ satisfying
\[
i_X\Omega = \dd H\quad\text{on }C_\infty.
\]
Any other solution differs from $X$ by a vector field taking values in 
$\ker \Omega|_{C_\infty}$.
\end{proposition}

\begin{proof}
The argument is the standard one in the GNH theory 
\cite{GotayNesterHinds1978,GotayNester1979}. 
Because $\Omega$ is only weakly non-degenerate in the infinite-dimensional 
setting, the equation $i_X\Omega=\dd H$ may fail to have solutions at some 
points of $W$ or may admit solutions that are not tangent to a given 
constraint submanifold. The GNH algorithm compensates for this by defining 
inductively a decreasing sequence
\[
C_0 \supseteq C_1 \supseteq C_2 \supseteq \cdots
\]
where $C_{k+1}$ consists of those points in $C_k$ at which the equation 
$i_X\Omega=\dd H$ admits a solution tangent to $C_k$.

The assumption that $\Omega$ and the 
compatibility conditions defining each $C_k$ have \emph{constant Banach 
rank} ensures that each $C_k$ is a (weak) submanifold of $W$. Under this 
regularity hypothesis, and assuming that the sequence stabilizes for some 
$k$,
\[
C_{k+1}=C_k,
\]
the presymplectic equation admits a solution tangent to $C_k$ at every 
point, by construction. Thus $C_k$ is the final constraint manifold 
$C_\infty$. Any two solutions differ by a vector field in 
$\ker(\Omega|_{C_\infty})$, since adding such a field does not change the 
interior product $i_X\Omega$.
\end{proof}

\begin{remark}
In the Banach or Hilbert setting, the constraint sets $C_k$ produced by 
the GNH algorithm are typically \emph{weak submanifolds} of~$W$. This means 
that each $C_k$ is modelled on a closed subspace of the model space of $W$, 
but not necessarily on a complemented one. Consequently, the inclusion 
$C_k\hookrightarrow W$ need not admit a continuous linear splitting, and 
tangent spaces satisfy $T_z C_k \subseteq T_z W$ as closed subspaces without requiring a topological direct sum 
decomposition.\hfill$\diamond$
\end{remark}

\begin{remark}
In models of continuum physics the GNH algorithm 
typically generates identities such as $\nabla\cdot B = 0$ and charge–Gauss 
        relations; compatibility conditions between particle velocities, 
        advected quantities, and electromagnetic variables; and 
dynamically generated constraints expressing the geometry of 
        field–velocity coupling. These constraints emerge automatically from the presymplectic geometry 
of~$(W,\Omega,H)$.\hfill$\diamond$
\end{remark}

\section{Lagrangian Formulation of Maxwell--Vlasov in the SR Framework}

We now reinterpret, in a unified way, the classical Low Lagrangian \cite{Low1958} for charged
particle dynamics and the geometric framework of Cendra--Holm--Hoyle--Marsden
(CHHM) \cite{CendraHolmHoyleMarsden1998} within the infinite-dimensional
Skinner--Rusk formalism. The Low Lagrangian provides a variational
formulation of the interaction between a distribution of charged particles
and the electromagnetic field and is usually written in Eulerian variables,
while the Skinner--Rusk formalism is naturally formulated on Lagrangian
configuration spaces. This mismatch makes it useful to relate the Eulerian
description $(F,\Phi,A)$ to its Lagrangian counterpart, in which particle
trajectories are encoded by a diffeomorphism of phase space. This transition
is precisely the geometric mechanism underlying the CHHM formulation and
will allow us to embed the Maxwell--Vlasov system into the SR framework.

In the Eulerian description, the basic
configuration variables are the distribution function $F(x,v)$ on phase
space and the electromagnetic potentials $(\Phi,A)$ on physical space. 
The particle velocity field $u$ does not appear as an independent
variable in the Eulerian Lagrangian: it is reconstructed from the
Lagrangian flow on phase space and enters only through the evolution
equation for $F$, not through the Lagrangian functional itself. This is
the geometric viewpoint emphasized in the CHHM framework
\cite{CendraHolmHoyleMarsden1998}.

The natural functional space of these fields is a Sobolev completion of
$(F,\Phi,A)$ in suitable $H^s$ spaces, which ensures that the required
operations---weak derivatives, curls, divergence and the pairing with the
distribution $F$---are well defined and that $(F,\Phi,A)$ form a Banach or
Hilbert manifold on which the SR formalism applies.

\begin{definition}
The \emph{Low Lagrangian} of the Maxwell--Vlasov system is the functional
\[
  L(\Phi,A,F;\dot A,\dot\Phi)
  =
  \int F(x,v)\left(\tfrac{1}{2}|v|^2 + q\Phi(x) - q\,v\!\cdot\!A(x)\right)\,dx\,dv
  +
  \frac{1}{2}
  \int \big(|\dot A + \nabla\Phi|^2 - |\nabla\times A|^2\big)\,dx .
\]
\end{definition}

This expression splits naturally into two contributions, corresponding
respectively to the particle and field sectors. 
The first integral represents the particle energy:
\begin{itemize}
  \item $\frac12 |v|^2$ is the kinetic energy per particle,
  \item $q\Phi$ is the electrostatic energy,
  \item $-q\,v\cdot A$ is the minimal electromagnetic coupling term.
\end{itemize}
The second integral is the standard electromagnetic field Lagrangian:
\[
  \mathbf{E} = -(\dot A + \nabla\Phi), 
  \qquad 
  \mathbf{B} = \nabla \times A,
\]
so that
\(
\frac12(|\dot A + \nabla\Phi|^2 - |\nabla\times A|^2)
= \tfrac12(|\mathbf{E}|^2 - |\mathbf{B}|^2).
\)

The Lagrangian depends on the time derivatives $(\dot A,\dot\Phi)$ but not
on $\dot F$, reflecting the fact that the Vlasov equation arises as a
kinematic (advective) constraint imposed by the particle flow rather than
from a kinetic term in the action. In particular, the Eulerian velocity $u$ of the phase-space flow is obtained from the
reconstruction relation $\dot\psi\,\psi^{-1} = u$, where $\psi$ is the
Lagrangian trajectory map in phase space. Thus $u$ is an
induced quantity rather than an independent variational variable.

\begin{proposition}
\label{prop:Low-degenerate}
The Low Lagrangian is degenerate in the directions of the scalar--potential 
velocity $\dot\Phi$ and in the particle--acceleration variables. In 
particular,
\[
  \frac{\delta L}{\delta \dot\Phi} = 0,
  \qquad 
  \frac{\delta L}{\delta \dot v} = 0,
\]
so the associated Legendre transform fails to be invertible and generates 
primary constraints in the sense of the GNH algorithm.
\end{proposition}

\begin{proof}
In the electromagnetic part of $L$, the scalar--potential velocity 
$\dot\Phi$ appears only through the gauge--invariant combination 
$\dot A + \nabla\Phi$, so the variational derivative of $L$ with respect 
to $\dot\Phi$ cannot produce an independent momentum variable. Therefore
\[
\frac{\delta L}{\delta \dot\Phi}=0,
\]
which is the familiar degeneracy leading to Gauss' law as a primary 
constraint.

For the particle sector, the Lagrangian depends on the phase--space 
velocity $v$ but contains no dependence on its time derivative 
$\dot v$. Hence the canonical momentum conjugate to $v$ vanishes:
\[
\frac{\delta L}{\delta \dot v}=0.
\]
This expresses the fact that particle accelerations are not variational 
degrees of freedom in the Eulerian formulation---the Vlasov equation 
arises as a kinematic constraint from the reconstruction of the 
phase--space flow rather than from a kinetic term in $L$. \end{proof}
\begin{remark}
In the Eulerian formulation adopted here, particle accelerations 
$\dot v$ do not appear as independent variables in the configuration 
manifold $Q$, so the corresponding momenta $P_v$ are not included 
explicitly in the Skinner--Rusk phase space $W$.  
The condition $\delta L/\delta \dot v = 0$ should therefore be interpreted 
as the Eulerian counterpart of the vanishing of the momenta conjugate to 
accelerations in the CHHM analysis \cite{CendraHolmHoyleMarsden1998}.  

In the CHHM framework these degeneracies are treated through a generalized 
Legendre transform and Dirac structures.  
In the Skinner--Rusk formalism, by contrast, they arise automatically as 
\emph{primary constraints} in the sense of the GNH algorithm.\hfill$\diamond$
\end{remark}

To express the Low Lagrangian in the Skinner--Rusk setting it is 
convenient to pass to the ``Lagrangian picture'' of the particle 
dynamics. The configuration space is taken to be $Q = \mathrm{Diff}(T\mathbb{R}^3)\times \mathcal{V} \times \mathcal{A}$, where $\psi\in\mathrm{Diff}(T\mathbb{R}^3)$ is the Lagrangian flow map on 
phase space, $\Phi\in \mathcal{V}$ is the scalar potential, and $A\in\mathcal{A}$ 
is the vector potential.

\begin{remark}
The natural Lagrangian variable for kinetic theories is the flow of
particles in phase space. This flow is a diffeomorphism 
$\psi : T\mathbb{R}^3 \to T\mathbb{R}^3$ describing the evolution of 
$(x,v)$ along exact particle trajectories, so $\mathrm{Diff}(T\mathbb{R}^3)$
plays the role of the Lagrangian configuration manifold for the Vlasov
sector. The relation between the Eulerian variables $(F,\Phi,A)$ and the
Lagrangian flow $\psi$ is encoded in the push-forward representation
$F(x,v,t) = \psi_t{}_*\!F_0$, where $F_0$ is the reference distribution. The Eulerian velocity field $u$
is then recovered from the reconstruction equation $\dot\psi\,\psi^{-1} = u$,
which expresses that $u$ is not an independent variational variable but
the instantaneous generator of the phase-space flow. In particular,
the Vlasov equation arises as a kinematic identity associated with this
reconstruction, rather than from a kinetic term in the Lagrangian.\hfill$\diamond$
\end{remark}

A point of the corresponding unified phase 
space $W$ is written as $(\psi,\Phi,A;\,\dot\psi,\dot\Phi,\dot A;\,P_\psi,P_\Phi,P_A)$. Next, we show how to derive the primary constraints for Maxwell–Vlasov systems.

\begin{proposition}
\label{prop:C0-MaxwellVlasov}
Let $L$ be the Low Lagrangian expressed in Lagrangian variables 
$(\psi,\Phi,A;\dot\psi,\dot\Phi,\dot A)$.  
The associated Skinner--Rusk generalized energy
\[
H =
  \langle P_\psi , \dot\psi\rangle
  + \langle P_\Phi , \dot\Phi\rangle
  + \langle P_A , \dot A\rangle
  - L(\psi,\Phi,A;\dot\psi,\dot\Phi,\dot A)
\]
defines a presymplectic system $(W,\Omega,H)$ whose primary constraint 
submanifold $C_0$ is given by
\[
  C_0
  =
  \bigl\{
    (\psi,\Phi,A;\dot\psi,\dot\Phi,\dot A;P_\psi,P_\Phi,P_A)\in W
    \;\big|\;
    P_\Phi = 0,\;
    P_\psi = \tfrac{\delta L}{\delta\dot\psi},\;
    P_A = \tfrac{\delta L}{\delta\dot A}
  \bigr\}.
\]
  
In particular, $C_0$ reproduces the generalized Legendre 
constraints obtained in the CHHM framework 
\cite{CendraHolmHoyleMarsden1998}.
\end{proposition}

\begin{proof}
By definition, a point of $W$ belongs to $C_0$ precisely when 
$\dd H(z)$ lies in the image of the presymplectic form $\Omega_z$.  
Since $\Omega$ pairs velocities with momenta, this requirement enforces 
the standard Legendre relations between the vertical variations of 
$H$ and the variational derivatives of $L$ with respect to the 
corresponding velocity components.

Evaluating $\dd H$ along variations of $(\dot\psi,\dot\Phi,\dot A)$ gives
\[
P_\psi = \frac{\delta L}{\delta\dot\psi},\qquad
P_\Phi = \frac{\delta L}{\delta\dot\Phi},\qquad
P_A    = \frac{\delta L}{\delta\dot A}.
\]
The Low Lagrangian satisfies 
$\delta L/\delta\dot\Phi = 0$ because $\dot\Phi$ enters only through 
$\dot A+\nabla\Phi$, yielding the primary constraint $P_\Phi=0$.  
Similarly, the particle contribution to $L$ contains no dependence on 
particle accelerations, so $\delta L/\delta\dot v = 0$, giving the 
constraint $P_v=0$.  

The remaining momenta $P_\psi$ and $P_A$ are determined by the 
non-degenerate part of the Legendre map. These relations are precisely 
the generalized Legendre constraints derived in the CHHM analysis 
\cite{CendraHolmHoyleMarsden1998}.
\end{proof}

\begin{remark}
The advantage of the SR perspective is that the constraints of Proposition~\ref{prop:C0-MaxwellVlasov} arise from a single presymplectic equation on $W$, without the need to separate Lagrangian and Hamiltonian variables \emph{a priori}.\hfill$\diamond$
\end{remark}
Applying the GNH algorithm to the presymplectic system 
\eqref{eq:SR-dynamics} on the primary constraint manifold $C_0$ produces,
in successive steps, the full geometric content of the Maxwell--Vlasov
system. The first iteration identifies the combinations of velocities and
potentials compatible with the kernel of $\Omega$, yielding the
reconstruction relation that identifies $\dot\psi$ with the Eulerian
phase-space velocity $u$ together with the gauge relation 
$E=-\,\dot A-\nabla\Phi$ defining the electric field. Subsequent
iterations enforce tangency of the dynamics to the evolving constraint
manifolds. In this process Gauss’ law appears in the form 
$\nabla\cdot E=\rho[F]$, where 
\[
\rho[F](x)=\int_{\mathbb{R}^3} F(x,v)\,dv
\]
is the charge density induced by the distribution function, while the
magnetic constraint $\nabla\cdot B=0$ arises as the condition that the
dynamics preserve the primary magnetic constraint.  
Tangency of the particle-flow constraints produces the Vlasov equation
for~$F$, and the compatibility of the electromagnetic momenta yields the
remaining Maxwell equations, including Ampère–Maxwell’s law in the form
\[
\partial_t E - \nabla\times B = -\,J[F],
\qquad 
J[F](x)=\int_{\mathbb{R}^3} v\,F(x,v)\,dv,
\]
where $J[F]$ is the current density carried by the particles.  
Thus every component of the Maxwell--Vlasov system emerges naturally as a
consistency condition uncovered by the GNH procedure.  
All of these constraints arise purely from the geometry of the SR
presymplectic system: none is imposed by hand, and each appears at a
finite stage of the GNH reduction.

\begin{remark}[Regularity of the SR Maxwell--Vlasov system]
In the finite-dimensional Skinner--Rusk theory 
\cite{SkinnerRusk1983,YoshimuraMarsden2006a,YoshimuraMarsden2006b},
the existence of solutions to the presymplectic equation 
$i_X\Omega = dH$ is guaranteed by a regularity condition: at each stage of 
the GNH algorithm, the constraint set must be a smooth submanifold and the 
restriction of the presymplectic form must have kernel of constant rank.  
This ensures that the algorithm stabilizes after a finite number of steps 
and that the resulting presymplectic dynamics is well posed.  

In the modern geometric formulation developed in 
\cite{ColomboMartindeDiegoZuccalli2010,ColomboDeDiego2014}, 
these regularity requirements are expressed in terms of weak submanifold 
theory on Banach or Hilbert manifolds and the constant-rank property for 
the Legendre map and the induced constraint distributions.

For the Maxwell--Vlasov system, these conditions are satisfied provided 
all fields and particle maps are taken in Sobolev spaces $H^s$ with 
$s > \tfrac{3}{2}+1$. Under this assumption:
\begin{enumerate}
\item the constraints 
$P_\Phi=0$, $P_v=0$, $E=-\dot A-\nabla \Phi$, 
$B=\nabla\times A$, 
$\nabla\cdot E=\rho[F]$, $\nabla\cdot B=0$
define smooth Banach submanifolds of $W$;
\item the kernel of $\Omega$ restricted to each constraint manifold 
$C_k$ has constant rank 
(particle--relabeling and gauge directions on $C_0$, 
pure gauge directions on $C_1$ and $C_\infty$);
\item the compatibility condition $dH(\ker\Omega)=0$ reproduces exactly 
the Gauss, magnetic, Vlasov, and Maxwell constraints, so no further 
obstructions to the existence of solutions appear.
\end{enumerate}

Therefore the Skinner--Rusk formulation of the Maxwell--Vlasov system is 
regular in the sense of 
\cite{ColomboMartindeDiegoZuccalli2010,ColomboDeDiego2014}, 
and the GNH sequence stabilizes after finitely many steps, yielding a 
well-defined presymplectic dynamical system on the final constraint 
manifold $C_\infty$.\hfill$\diamond$
\end{remark}

\begin{theorem}[SR formulation of the Maxwell--Vlasov equations]
\label{thm:SR-gives-MV}
Let $L$ be the Low Lagrangian on 
$Q=\mathrm{Diff}(T\mathbb{R}^3)\times \mathcal{V}\times\mathcal A$, and let 
$(W,\Omega,H)$ be the associated Skinner--Rusk system.  
Then the GNH algorithm stabilizes at a final constraint manifold 
$C_\infty\subset W$ such that the restriction of the SR equation 
\eqref{eq:SR-dynamics} to $C_\infty$ is equivalent to the full set of 
Maxwell--Vlasov equations:
\begin{align}
\partial_t F + \{F,H_{\mathrm{part}}\} &= 0,\\
\partial_t E - \nabla\times B &= -J[F],\\
\partial_t B + \nabla\times E &= 0,\\
\nabla\cdot E &= \rho[F],\qquad \nabla\cdot B = 0,
\end{align}
where the charge and current densities are
\[
\rho[F](x)=\int_{\mathbb{R}^3} F(x,v)\,dv,
\qquad
J[F](x)=\int_{\mathbb{R}^3} v\,F(x,v)\,dv,
\]
and $\{\cdot,\cdot\}$ is the canonical Poisson bracket on phase space,
\[
\{F,H_{\mathrm{part}}\}(x,v)
=
\nabla_x F\cdot\nabla_v H_{\mathrm{part}}
-
\nabla_v F\cdot\nabla_x H_{\mathrm{part}}.
\]
The single–particle Hamiltonian associated with the Low Lagrangian is
\[
H_{\mathrm{part}}(x,v;\Phi,A)
=
\frac{1}{2}|v|^2 + q\,\Phi(x) - q\,v\!\cdot\!A(x).
\]
\end{theorem}

\begin{proof}
We work on the phase space $W = TQ \oplus T^*Q$ with coordinates
$(q,v,p) = (\psi,\Phi,A;\dot\psi,\dot\Phi,\dot A;P_\psi,P_\Phi,P_A)$. By definition, the presymplectic form is the pullback of the canonical
symplectic form on $T^*Q$, $\Omega = \mathrm{pr}_2^*\omega_{\mathrm{can}}$, so for tangent vectors 
$Y_1=(\delta q_1,\delta v_1,\delta p_1)$ and 
$Y_2=(\delta q_2,\delta v_2,\delta p_2)$ one has $\Omega(Y_1,Y_2)
=
\langle \delta q_1,\delta p_2\rangle
-
\langle \delta q_2,\delta p_1\rangle$. In particular, if 
$X=(\dot q,\dot v,\dot p)$ is a vector field on $W$, the condition
$i_X\Omega = \dd H$ is equivalent to
\begin{equation}
\label{eq:SR-local}
\Omega(X,Y) = \dd H(Y)
\qquad\text{for all tangent vectors }Y=(\delta q,\delta v,\delta p).
\end{equation}

The Skinner--Rusk generalized energy is $H(q,v,p) = \langle p,v\rangle - L(q,v)$, so its differential along $Y=(\delta q,\delta v,\delta p)$ is $\dd H(Y)
=
\langle \delta p, v\rangle
+
\langle p,\delta v\rangle
-
\Big\langle \frac{\delta L}{\delta q},\delta q\Big\rangle
-
\Big\langle \frac{\delta L}{\delta v},\delta v\Big\rangle$. On the other hand, $\Omega(X,Y)
=
\langle \dot q,\delta p\rangle
-
\langle \delta q,\dot p\rangle$. Imposing \eqref{eq:SR-local} for arbitrary $(\delta q,\delta v,\delta p)$
gives, by comparing the coefficients of $\delta p$, $\delta v$ and $\delta
q$:
\begin{align}
\dot q &= v, \label{eq:SR-q-evolution}\\
p     &= \frac{\delta L}{\delta v}, \label{eq:SR-momenta}\\
\dot p &= \frac{\delta L}{\delta q}. \label{eq:SR-Euler-Lagrange}
\end{align}
Equations \eqref{eq:SR-q-evolution}--\eqref{eq:SR-Euler-Lagrange} are the
unified SR form of the Euler--Lagrange equations for the Lagrangian $L$,
together with the (generalized) Legendre relations.

For the Low Lagrangian on 
$Q=\mathrm{Diff}(T\mathbb{R}^3)\times \mathcal{V}\times\mathcal A$ we have seen in
Proposition~\ref{prop:Low-degenerate} that
\[
\frac{\delta L}{\delta \dot\Phi}=0,
\qquad
\frac{\delta L}{\delta \dot v}=0,
\]
so the corresponding momenta $P_\Phi$ and $P_v$ vanish on the primary
constraint manifold. This is precisely the description of $C_0$ given in
Proposition~\ref{prop:C0-MaxwellVlasov}. Restricting 
\eqref{eq:SR-q-evolution}--\eqref{eq:SR-Euler-Lagrange} to $C_0$ yields
the evolution equations for $(\psi,\Phi,A)$ together with the relations
between $(P_\psi,P_A)$ and $(\dot\psi,\dot A)$.

On $C_0$ the kernel of $\Omega$ is spanned by infinitesimal
particle-relabeling transformations (acting on $\psi$) and by gauge
transformations of the potentials $(\Phi,A)$. Requiring that $\dd H$
annihilate this kernel gives the first set of compatibility conditions:
the reconstruction relation $\dot\psi\,\psi^{-1}=u$, which identifies
$\dot\psi$ with the Eulerian phase-space velocity field $u$, and the gauge
relation $E=-\,\dot A-\nabla\Phi$, which expresses the electric field in
terms of $(\Phi,A)$. These conditions define the first constraint
manifold $C_1$.

The second step for the GNH algorithm requires that the vector field $X$ be tangent to
$C_1$. Tangency to the reconstruction relation implies that the push-forward
representation $F=\psi_*F_0$ satisfies
\[
\partial_t F + \nabla_x\cdot(F u_x) + \nabla_v\cdot(F u_v) = 0,
\]
where $(u_x,u_v)$ is the phase-space velocity induced by the particle
Hamiltonian $H_{\mathrm{part}}(x,v;\Phi,A)
= \frac{1}{2}|v|^2 + q\,\Phi(x) - q\,v\cdot A(x)$ through the canonical
Hamilton equations
\[
u_x = \nabla_v H_{\mathrm{part}},\qquad
u_v = -\,\nabla_x H_{\mathrm{part}}.
\]
In Hamiltonian form this is exactly
\[
\partial_t F + \{F,H_{\mathrm{part}}\} = 0.
\]

Tangency to the electromagnetic constraint $E=-\,\dot A-\nabla\Phi$ and to
the magnetic relation $B=\nabla\times A$ yields, using
\eqref{eq:SR-Euler-Lagrange} for the electromagnetic momenta,
\[
\partial_t B + \nabla\times E = 0,
\qquad
\partial_t E - \nabla\times B = -\,J[F],
\]
where $J[F](x)=\int v F(x,v)\,dv$ is the current density appearing in the
source terms of the Maxwell equations. Preservation in time of the
constraints generated by the degeneracy in $\dot\Phi$ gives Gauss' law
\[
\nabla\cdot E = \rho[F],\qquad
\rho[F](x)=\int F(x,v)\,dv,
\]
while preservation of the magnetic constraint yields $\nabla\cdot B=0$.

No new compatibility conditions arise at higher iterations: once the
Vlasov equation, Maxwell equations and Gauss/divergence constraints are
imposed, the dynamics automatically preserves all constraints. Therefore
the GNH sequence stabilizes at $C_\infty=C_2$, and the restriction of the
SR equation $i_X\Omega=\dd H$ to $C_\infty$ is equivalent to the full
Maxwell--Vlasov system written in terms of $(F,E,B)$. This recovers, in
the SR/GNH setting, the constraint structure identified in the Dirac and
CHHM analyses \cite{CendraHolmHoyleMarsden1998}.
\end{proof}

\begin{remark}
In the GNH construction underlying Theorem~\ref{thm:SR-gives-MV}, the 
constraint manifolds form a finite descending chain
\[
C_0 \supset C_1 \supset C_2 = C_\infty,
\]
where $C_0$ encodes the Legendre degeneracy and gauge constraints, $C_1$
implements the reconstruction and electromagnetic relations, and $C_2$
imposes Gauss' law, the magnetic constraint and the Vlasov equation. This
hierarchy provides a clean geometric separation between purely kinematic
constraints and dynamically generated ones in the Maxwell--Vlasov system.\hfill$\diamond$
\end{remark}

The preceding analysis shows that the Low Lagrangian, when embedded in the
Skinner--Rusk framework and treated with the GNH algorithm, reproduces the
full Maxwell--Vlasov system on a final constraint manifold $C_\infty$
equipped with a natural presymplectic structure. In the next section we
exploit the invariance of this construction under particle-relabeling
transformations to perform a symmetry reduction by 
$\mathrm{Diff}(T\mathbb{R}^3)$. This yields, in a unified way, both the
Euler--Poincaré equations for the Vlasov sector and the Lie--Poisson
Hamiltonian structure underlying the classical Maxwell--Vlasov brackets.

\section{Reduction of the Skinner--Rusk Maxwell--Vlasov System}

We now describe how the Skinner--Rusk (SR) formulation of the 
Maxwell--Vlasov system admits a natural symmetry reduction under the 
action of the phase--space diffeomorphism group.  
Unlike the finite--dimensional setting, reduction must be performed in the 
category of weak Banach (or Hilbert) manifolds, and particular care is 
required regarding smoothness and regularity of the group action.  
The construction below follows the functional–analytic framework developed
in \cite{EbinMarsden1970,CendraHolmHoyleMarsden1998}.

Let $s>\frac32+1$ and consider the Sobolev manifold $G = \Diff^s(T\mathbb{R}^3)$, the group of $H^s$–diffeomorphisms of phase space.   By the classical Ebin–Marsden theorem \cite{EbinMarsden1970}, $G$ is a 
smooth Hilbert manifold and a topological group; composition and 
inversion are $C^\infty$ maps in the Hilbert sense.  
We take the configuration manifold
\[
Q = \Diff^s(T\mathbb{R}^3)\times \mathcal{V}^s\times \mathcal{A}^s,
\]
where $(\Phi,A)\in(\mathcal V^s,\mathcal A^s)$ are Sobolev potentials.

The action of $G$ on $Q$ is the relabeling (right) action
\[
g\cdot(\psi,\Phi,A) = (\psi\circ g^{-1},\, \Phi,\, A),
\]
which is the natural one for Vlasov-type systems and matches the
CHHM convention \cite{CendraHolmHoyleMarsden1998}.

This action is smooth, free and admits local slices, hence it is 
weakly proper in the sense of \cite{EbinMarsden1970}.  
It lifts smoothly to $TQ$, $T^*Q$, and therefore to the weak 
manifold $W = TQ\oplus T^*Q$.

\begin{remark}
In Sections~2--4 we work formally with smooth potentials 
$(\Phi,A)\in\mathcal{V}\times\mathcal{A}$, where 
$\mathcal{V}$ denotes the space of scalar potentials and 
$\mathcal{A}$ the space of vector potentials.  
For the rigorous analytic construction in the SR--GNH framework, 
and following the standard approach of Ebin--Marsden 
\cite{EbinMarsden1970}, we later replace these spaces by their Sobolev 
counterparts $\mathcal{V}^s$ and $\mathcal{A}^s$ with $s>\frac32+1$.  
This ensures that the particle--relabeling group 
$\Diff^s(T\mathbb{R}^3)$ is a smooth Hilbert manifold and that all 
maps in the SR construction (including the Low Lagrangian, the canonical 
pairings, and the presymplectic form) are well-defined and $C^\infty$ in 
the Hilbert sense.  No conceptual change is introduced: the Sobolev 
setting is merely a technically convenient completion of the smooth 
configuration manifold.
\hfill$\diamond$
\end{remark}

\begin{lemma}\label{prop:SR-invariance}
Let $G = \Diff^s(T\mathbb{R}^3)$ acting on the configuration manifold 
$Q = \Diff^s(T\mathbb{R}^3)\times \mathcal V^s\times\mathcal A^s$.
 by 
$g\cdot(\psi,\Phi,A) = (\psi\circ g^{-1},\;\Phi,\;A)$ for $g\in G$,
and let this action lift to $TQ$ and $T^\ast Q$ by the canonical tangent
and cotangent lifts. 

Let $L:TQ\to\mathbb{R}$ be the Low Lagrangian, written in Lagrangian variables
$(\psi,\Phi,A;\dot\psi,\dot\Phi,\dot A)$, and assume $s>\tfrac32+1$ so
that all spaces are smooth Hilbert manifolds and the $G$-action is smooth,
free and proper. If $L$ is $G$--invariant when expressed in terms of the advected density
$F=\psi_*F_0$ and the electromagnetic potentials $(\Phi,A)$, then:
\begin{enumerate}
\item The Hamiltonian $H(q,v,p)=\langle p,v\rangle - L(q,v),\, (q,v,p)\in W$, is $G$-invariant as a map $H:W\to\mathbb{R}$.
\item The weak presymplectic form 
$\Omega=\operatorname{pr}_2^\ast\omega_{\mathrm{can}}$ on $W$ is $G$-invariant.  
Moreover, if $J:T^\ast Q\to\mathfrak g^\ast$ denotes the momentum map of
the cotangent-lifted action, then for every $\xi\in\mathfrak g$ one has
\[
\iota_{\xi_W}\Omega
=
\dd\big( \langle J\circ\operatorname{pr}_2,\xi\rangle \big),
\]
where $\xi_W$ is the fundamental vector field on $W$ generated by $\xi$.
In particular, the restriction of $\Omega$ to any momentum level set
$(J\circ\operatorname{pr}_2)^{-1}(\mu)\subset W$ is horizontal with respect to the
$G$-action and hence defines a basic weak two-form on the reduced space.
\end{enumerate}
Consequently, the system $(W,\Omega,H)$ admits a well-defined
symmetry reduction by $G$ in the category of weak Hilbert manifolds, in the
sense of Marsden--Weinstein–type reduction on the appropriate momentum
level sets.
\end{lemma}

\begin{proof}

By hypothesis, $L$ depends on $\psi$ only through the advected density
$F=\psi_*F_0$ and on $(\Phi,A)$ as fields on physical space. The action of
$G$ on $\psi$ is by relabeling of the reference phase space $(g\cdot\psi)(z) = \psi(g^{-1}z)$, with $z\in T\mathbb{R}^3$.

If $F_0$ is chosen to be $G$-invariant (as in the usual CHHM setting \cite{CendraHolmHoyleMarsden1998},
e.g.\ a reference density on phase space), then
\[
(g\cdot\psi)_*F_0 = \psi_* (g^{-1}_*F_0) = \psi_*F_0 = F,
\]
so the advected density is unchanged. Since the action leaves $(\Phi,A)$
inert, it follows that $L$ is $G$-invariant, that is, 
\[
L\bigl(g\cdot(q,v)\bigr) = L(q,v),\qquad (q,v)\in TQ,\ g\in G.
\]

The action of $G$ on $T^\ast Q$ is the canonical cotangent lift of the
action on $Q$ given by $g\cdot(q,p) = \bigl(g\cdot q,\ (T_q g^{-1})^\ast p\bigr)$. The canonical pairing satisfies $\langle (T_q g^{-1})^\ast p,\; T_q g\cdot v\rangle 
= \langle p, v\rangle$, so the term $\langle p,v\rangle$ is $G$-invariant. Since $L$ is
$G$-invariant, the Hamiltonian
$H(q,v,p)=\langle p,v\rangle-L(q,v)$ is also invariant, that is,  $H\bigl(g\cdot(q,v,p)\bigr) = H(q,v,p)$ $\forall\,g\in G$.

Next, note that on $T^\ast Q$ the cotangent-lifted $G$-action is symplectic, i.e.
\[
(g^\ast\omega_{\mathrm{can}})_{(q,p)} = \omega_{\mathrm{can}}{}_{(q,p)}
\qquad\forall\,g\in G,
\]
so $\omega_{\mathrm{can}}$ is $G$-invariant. The fundamental vector field
$\xi_{T^\ast Q}$ associated with $\xi\in\mathfrak g$ is Hamiltonian with
Hamiltonian function $\langle J,\xi\rangle$, where $J:T^\ast Q\to\mathfrak g^\ast$, with 
$\langle J(q,p),\xi\rangle = \langle p, \xi_Q(q)\rangle$
is the standard momentum map. Thus
\[
\iota_{\xi_{T^\ast Q}}\omega_{\mathrm{can}}
=
\dd\langle J,\xi\rangle.
\]

On the space $W$ we consider the product
action $g\cdot(q,v,p) = (g\cdot q,\;g\cdot v,\;g\cdot p)$, whose projection on the second factor is exactly the cotangent lift on
$T^\ast Q$. The presymplectic form is $\Omega=\operatorname{pr}_2^\ast\omega_{\mathrm{can}}$,
so for the fundamental vector field $\xi_W$ induced by $\xi\in\mathfrak g$
we have
\[
\iota_{\xi_W}\Omega
=
\iota_{\xi_W}\operatorname{pr}_2^\ast\omega_{\mathrm{can}}
=
\operatorname{pr}_2^\ast\bigl(\iota_{\xi_{T^\ast Q}}\omega_{\mathrm{can}}\bigr)
=
\operatorname{pr}_2^\ast\dd\langle J,\xi\rangle
=
\dd\bigl(\langle J\circ\operatorname{pr}_2,\xi\rangle\bigr).
\]
In particular, $\Omega$ is $G$-invariant (as a pullback of a
$G$--invariant form) and its contraction with any vertical vector 
$\xi_W$ is exact. Therefore, on each momentum level set
\[
\mathcal C_\mu := (J\circ\operatorname{pr}_2)^{-1}(\mu)\subset W,
\]
we have $\iota_{\xi_W}\Omega|_{\mathcal C_\mu}=0$, so the restriction
$\Omega|_{\mathcal C_\mu}$ is horizontal and $G$-invariant. That is,
$\Omega|_{\mathcal C_\mu}$ is a basic weak two-form and descends to a
well-defined presymplectic form on the reduced space 
$\mathcal C_\mu/G$ in the sense of 
Marsden--Weinstein reduction \cite{MarsdenWeinstein1974}, \cite{OrtegaRatiu2004}.
\end{proof}

For each value $\mu\in\mathfrak g^\ast$, consider the momentum level set $\mathcal C_\mu := (J\circ\operatorname{pr}_2)^{-1}(\mu)\subset W$. By Lemma~\ref{prop:SR-invariance}, the restriction of $\Omega$ to 
$\mathcal C_\mu$ is $G$-invariant and horizontal, so it is a basic weak
two-form. Since the $G=\Diff^s(T\mathbb{R}^3)$ action is free, proper, and smooth for $s>\tfrac32+1$
\cite{EbinMarsden1970,OrtegaRatiu2004}, the quotient
\[
W_{\mathrm{red}}^\mu := \mathcal C_\mu/G
\]
is a weak Hilbert manifold. Let $\pi_\mu : \mathcal C_\mu \to W_{\mathrm{red}}^\mu$ denote the orbit projection; it is a smooth weak submersion.

A choice of local section of the bundle $Q \to Q/G$ allows us to introduce
explicit reduced coordinates.  
Since the relabeling group is $G=\Diff^s(T\mathbb{R}^3)$, with
$s>\tfrac32+1$, its Lie algebra is $\mathfrak g^s = \mathfrak X^s(T\mathbb{R}^3)$, the Sobolev $H^s$ space of vector fields on phase space.
The dual space $(\mathfrak g^s)^\ast$ is identified, via the $L^2$ pairing,
\[
\langle \mu , u\rangle
  = \int_{T\mathbb{R}^3} \mu\cdot u \, dx\,dv,
  \qquad u\in\mathfrak g^s,
\]
with the space of $H^{s-1}$ one–form densities on $T\mathbb{R}^3$,  $(\mathfrak g^s)^\ast 
\simeq 
H^{s-1}\!\left(T^\ast(T\mathbb{R}^3)\otimes \mathrm{Dens}(T\mathbb{R}^3)\right)$.

For notational simplicity, we fix a momentum value $\mu$ (typically $\mu=0$
for the particle-relabeling symmetry) and denote the reduced space
$W_{\mathrm{red}}^\mu$ simply by $W_{\mathrm{red}}$, and the projection
$\pi_\mu$ by $\pi_W$.

Under these identifications, the orbit projection 
$\pi_W \equiv \pi_\mu : \mathcal C_\mu \to W_{\mathrm{red}}$ 
provides reduced coordinates $(u,F,A,\Phi;\mu,\alpha)$, where 
$u = \dot\psi \circ \psi^{-1} \in \mathfrak g^s$ is the Eulerian 
phase–space velocity obtained from the Lagrangian 
velocity $\dot\psi$ via right translation on $\Diff^s(T\mathbb{R}^3)$; 
$F = \psi_* F_0$ is the advected distribution function (push–forward of 
the reference density $F_0$); $\mu \in (\mathfrak g^s)^\ast$ is the reduced 
momentum obtained from the covariant momentum $P_\psi$ by
\[
\langle \mu , \delta\psi\,\psi^{-1} \rangle
  = \langle P_\psi , \delta\psi \rangle,
\]
so that, in coordinates, $\mu = P_\psi\circ T\psi$; and $(A,\Phi)$ are the 
electromagnetic potentials, and 
$\alpha=(P_A,P_\Phi)\in T^\ast(\mathcal A\times\mathcal V)$ are their
canonical momenta. These variables are inert under the $G$–action.

With these conventions, the projection $\pi_W$ takes the explicit form
\[
\pi_W(\psi,\Phi,A;\dot\psi,\dot\Phi,\dot A;P_\psi,P_\Phi,P_A)
=
\bigl( u , F , A , \Phi ; \mu , \alpha \bigr),
\]
where the reduced variables $(u,F,\mu)$ encode the geometric content of
particle–relabeling symmetry, and $(A,\Phi;\alpha)$ describe the
electromagnetic sector unaffected by the action of $G$.

\begin{proposition}\label{prop:Omega-red}
Under the hypotheses of Lemma~\ref{prop:SR-invariance},
the restriction of $\Omega$ to each momentum level set 
$\mathcal C_\mu$ is basic with respect to the $G$-action.
Consequently, for each $\mu\in\mathfrak g^\ast$ there exists a unique weak 
presymplectic form $\Omega_{\mathrm{red}}^\mu$ on $W_{\mathrm{red}}^\mu$
such that
\begin{equation}
\label{eq:Omega-red-def}
\pi_\mu^\ast \Omega_{\mathrm{red}}^\mu = \Omega|_{\mathcal C_\mu}.
\end{equation}
Fixing $\mu$ and writing $W_{\mathrm{red}}$ and 
$\Omega_{\mathrm{red}}$ for $W_{\mathrm{red}}^\mu$ and $\Omega_{\mathrm{red}}^\mu$,
respectively, we obtain in reduced coordinates $(u,F,A,\Phi;\mu,\alpha)$
the splitting
\[
\Omega_{\mathrm{red}} 
  = \Omega_{\mathrm{can}}
  \;+\;
    \Omega_{\mathrm{LP}},
\] where $\Omega_{\mathrm{can}}$ is the canonical contribution from the
electromagnetic potentials and 
$\Omega_{\mathrm{LP}}$ is the Lie--Poisson part:
\[
\Omega_{\mathrm{LP}}
\bigl(
(\delta u_1,\delta\mu_1),(\delta u_2,\delta\mu_2)
\bigr)
=
\langle \delta u_1, \delta \mu_2\rangle 
- \langle \delta u_2, \delta \mu_1\rangle
+ \langle \mu , [\delta u_1,\delta u_2] \rangle.
\]
\end{proposition}

\begin{proof}
Lemma~\ref{prop:SR-invariance} shows that 
$\Omega|_{\mathcal C_\mu}$ is $G$–invariant and that $\forall\,\xi\in\mathfrak g$
\[
\iota_{\xi_W}\Omega = \dd\langle J\circ\operatorname{pr}_2 , \xi\rangle.
\]
Since $\mathcal C_\mu = (J\circ\operatorname{pr}_2)^{-1}(\mu)$ is a level
set of the momentum map, the function
$\langle J\circ\operatorname{pr}_2 , \xi\rangle$
is constant on $\mathcal C_\mu$. Therefore, $\dd\langle J\circ\operatorname{pr}_2 , \xi\rangle|_{\mathcal C_\mu} = 0$, and hence $\forall\,\xi\in\mathfrak g$
\[
\iota_{\xi_W}\Omega|_{\mathcal C_\mu} = 0,
\]
showing that $\Omega|_{\mathcal C_\mu}$ annihilates all vertical vectors
of the $G$–action. Combined with $G$–invariance, this proves that
$\Omega|_{\mathcal C_\mu}$ is a \emph{basic} weak two–form in the sense of
the Hilbert–Sobolev setting of Ebin–Marsden.

Since the action is smooth, free and proper,
$\pi_\mu:\mathcal C_\mu\to W_{\mathrm{red}}^\mu$ is a smooth principal
$G$–bundle and, by the standard result on basic forms
(\cite[Thm.~2.3.4]{OrtegaRatiu2004}), there exists a unique weak
two–form $\Omega_{\mathrm{red}}^\mu$ on $W_{\mathrm{red}}^\mu$ satisfying
\[
\pi_\mu^\ast\Omega_{\mathrm{red}}^\mu
  = \Omega|_{\mathcal C_\mu}.
\]

To compute its expression, we use the identifications
$T\Diff^s/G\simeq\mathfrak g^s$ and $T^\ast\Diff^s/G\simeq(\mathfrak g^s)^\ast$.
Under these identifications the bundle $W_{\mathrm{red}}^\mu$ splits as
\[
W_{\mathrm{red}}^\mu
\simeq
\bigl(\mathfrak g^s\times(\mathfrak g^s)^\ast\bigr)
  \times
T^\ast(\mathcal A^s\times\mathcal V^s).
\]
On the electromagnetic sector, $\Omega$ coincides with the canonical
symplectic form on $T^\ast(\mathcal A^s\times\mathcal V^s)$, hence yields
$\Omega_{\mathrm{can}}$ on the quotient.

On the $(u,\mu)$ sector, the reduced form is determined by the canonical
pairing and the coadjoint action:
for variations $(\delta u_1,\delta\mu_1)$ and $(\delta u_2,\delta\mu_2)$,
\[
\Omega_{\mathrm{LP}}\big((\delta u_1,\delta\mu_1),(\delta u_2,\delta\mu_2)\big)
=
\langle \delta u_1,\delta\mu_2\rangle 
 -\langle \delta u_2,\delta\mu_1\rangle
 +\langle \mu,[\delta u_1,\delta u_2]\rangle,
\]
which is the standard Lie--Poisson 2–form on $(\mathfrak g^s)^\ast$. Thus $\Omega_{\mathrm{red}}^\mu = \Omega_{\mathrm{can}} + \Omega_{\mathrm{LP}}$.
\end{proof}


\subsection{The Reduced Equations: Lagrangian picture}

Having constructed the reduced presymplectic manifold 
$(W_{\mathrm{red}},\Omega_{\mathrm{red}})$, we now derive the unified 
Eulerian equations of motion obtained by applying the GNH algorithm to the
triple $(W_{\mathrm{red}},\Omega_{\mathrm{red}},H_{\mathrm{red}})$.

Recall that the unified Skinner--Rusk Hamiltonian on 
$W$ descends to a reduced Hamiltonian $H_{\mathrm{red}} : W_{\mathrm{red}} \to\mathbb{R}$ satisfying $\pi_W^\ast H_{\mathrm{red}} = H|_{\mathcal C_\mu}$.
In the reduced coordinates 
$(u,F,A,\Phi;\mu,\alpha)$ introduced earlier, one obtains the explicit 
expression
\begin{equation}\label{eq:Hred-def}
H_{\mathrm{red}}(u,F,A,\Phi;\mu,\alpha)
  = \langle \mu,u\rangle 
  + \langle \alpha , (A,\Phi) \rangle 
  - l(u,F,A,\Phi),
\end{equation}
where $l$ is the reduced Lagrangian of 
\cite{CendraHolmHoyleMarsden1998} obtained by expressing the Low 
Lagrangian in Eulerian variables and eliminating the redundant momenta 
associated with particle–relabeling symmetry, given by
\begin{equation}\label{eq:l-reduced}
l(u,F,A,\Phi)
=
\int_{T\mathbb{R}^3}
F(x,v,t)\,
\big( \tfrac12 |v|^2 + v\cdot A(x,t) - \Phi(x,t) \big)
\,dx\,dv
\;-\;
\frac{1}{2}\int_{\mathbb{R}^3}
\big( |E|^2 + |B|^2 \big)\,dx,
\end{equation} where
\[
E = -\,\partial_t A \;-\; \nabla\Phi,
\qquad
B = \nabla\times A,
\]
and the Eulerian velocity field $u$ enters implicitly through the advection
of the distribution function $F$, namely $\partial_t F + \mathcal{L}_u F = 0$.

The presymplectic dynamics in the Skinner--Rusk formalism is encoded by
the equation
\[
\iota_X \Omega_{\mathrm{red}} = \dd H_{\mathrm{red}},
\]
to be solved on the \emph{final reduced constraint manifold} 
$C_{\infty,\mathrm{red}} \subset W_{\mathrm{red}}$, obtained by applying 
the GNH algorithm to the reduced system.  

Since the $G$–action preserves both the Skinner--Rusk two–form $\Omega$
and the Hamiltonian $H$, the GNH algorithm is $G$–equivariant at every
stage: each constraint manifold $C_k\subset W$ is $G$–invariant, and the
corresponding tangency conditions descend to the quotient.  
Therefore, by the general result on equivariant presymplectic reduction
(see \cite{OrtegaRatiu2004}), performing Marsden--Weinstein reduction and
applying the GNH algorithm commute.  

In particular,
\[
C_{\infty,\mathrm{red}}
\;=\;
C_\infty / G,
\]
so the final reduced constraint manifold is precisely the image of the
unreduced final constraint manifold under the quotient by $G$.

We now show that the reduced unified equation reproduces the full 
Euler--Poincaré system with advected density $F$, coupled to the Maxwell 
equations for the electromagnetic fields.

\begin{theorem}
\label{thm:EP-from-SR}
Let $(W_{\mathrm{red}},\Omega_{\mathrm{red}},H_{\mathrm{red}})$ be the 
reduced Skinner--Rusk system associated with the Maxwell--Vlasov 
Lagrangian, and let $C_{\infty,\mathrm{red}}$ be its final reduced 
constraint manifold obtained by the GNH algorithm.  
Then the following statements hold:

\begin{enumerate}
\item[(1)] 
The presymplectic equation
\[
\iota_X \Omega_{\mathrm{red}} = \dd H_{\mathrm{red}}
\qquad\text{on }C_{\infty,\mathrm{red}}
\]
admits solutions $X$ that are tangent to $C_{\infty,\mathrm{red}}$.

\item[(2)]
The $u$–equation obtained from the reduced Skinner--Rusk equation
$\iota_X\Omega_{\mathrm{red}}=\dd H_{\mathrm{red}}$
is the Euler--Poincaré equation on $\mathfrak g^\ast$ with advected
density $F$, 

\[
\frac{\mathrm d}{\mathrm dt}\!\left(\frac{\delta l}{\delta u}\right)
\;+\;
\operatorname{ad}^\ast_{u}
\!\left(\frac{\delta l}{\delta u}\right)
\;=\;
F \diamond \left(\frac{\delta l}{\delta F}\right)
\;+\;
\mathcal{F}_{\mathrm{EM}}.
\]

Here the \emph{diamond operator} associated to the action of
$\Diff(T\mathbb{R}^3)$ on the advected density $F$ is defined by $\big\langle F \diamond \eta,\, \xi \big\rangle
\;:=\;
\big\langle F,\, -\mathcal L_\xi \eta \big\rangle,
\,\,
\xi\in\mathfrak g,\;
\eta\in T_F^\ast\mathrm{Dens}(T\mathbb{R}^3)$, where $\mathcal L_\xi$ denotes the Lie derivative.

The term $\mathcal{F}_{\mathrm{EM}}$ is the electromagnetic forcing induced
by the Low Lagrangian.  In phase–space variables 
$(x,v)\in T\mathbb{R}^3$, it is given by
\[
\mathcal{F}_{\mathrm{EM}}(x,v)
=
F(x,v)
\Bigl(
\, q\,E(x)
\;+\;
q\, v \times B(x)
\Bigr),
\]
where
\[
E = -\partial_t A - \nabla\Phi,
\qquad
B = \nabla\times A,
\]
are the electromagnetic fields, so that the force term in the
Euler--Poincaré equation corresponds to the Lorentz force density
$q(E + v\times B)$ weighted by the distribution function~$F$.

\item[(3)]
The equation obtained by varying $H_{\mathrm{red}}$ with respect to $F$
is the advection equation
\[
\partial_t F + \mathcal{L}_u F = 0,
\]
encoding the push-forward of the distribution function by the Eulerian 
velocity field $u$.

\item[(4)]
The $(A,\Phi)$--sector reproduces Maxwell’s equations in Eulerian form:
\[
\partial_t (E,B)
  \;=\;
\text{Maxwell equations with current }j(F,u),
\]
\[
\begin{cases}
\partial_t B + \nabla\times E = 0, \\[4pt]
\partial_t E - \nabla\times B = -\, j(F,u), \\[4pt]
\nabla\cdot B = 0, \\[4pt]
\nabla\cdot E = \rho(F),
\end{cases}
\]
where the charge and current densities induced by the distribution
function are
\[
\rho(F)(x)
  = \displaystyle\int_{\mathbb{R}^3} F(x,v)\, dv,
\qquad
j(F,u)(x)
  = \displaystyle\int_{\mathbb{R}^3} u_x(x,v)\, F(x,v)\, dv.
\]
Here \(u_x(x,v)=v\) is the spatial component of the Eulerian
phase--space velocity field \(u(x,v)= (u_x,u_v)\). In particular, the constraints 
$\nabla\!\cdot B = 0$ and 
$\nabla\!\cdot E = \rho[F]$ are preserved by the flow.
\end{enumerate}

Thus the reduced Skinner--Rusk dynamics is equivalent to the full 
Euler--Poincaré–Maxwell system for the Maxwell--Vlasov plasma.
\end{theorem}

\begin{proof}
We prove each item in order, using the reduced presymplectic system 
$(W_{\mathrm{red}},\Omega_{\mathrm{red}},H_{\mathrm{red}})$ and the 
structure of the Lie--Poisson/canonical splitting.

(1) By Lemma~\ref{prop:SR-invariance}, $\Omega_{\mathrm{red}}$ is a basic weak 
two--form obtained from $\Omega$ by reduction.  
Since reduction commutes with the GNH algorithm, one has the equality of final 
constraint sets $C_{\infty,\mathrm{red}} \;=\; C_\infty/G$, where $C_\infty$ is the final constraint manifold for the unreduced 
Skinner--Rusk system.  
Because the unreduced equation 
$\iota_X\Omega=\dd H$ is solvable on $C_\infty$ by the general theory of 
presymplectic GNH systems, its pushforward under the quotient map yields a 
solution to $\iota_X \Omega_{\mathrm{red}} = \dd H_{\mathrm{red}}$ that is tangent to $C_{\infty,\mathrm{red}}$.  
Thus solvability on the reduced final constraint set is guaranteed.

(2) The reduced two--form splits as $\Omega_{\mathrm{red}} = \Omega_{\mathrm{can}} + \Omega_{\mathrm{LP}}$, where $\Omega_{\mathrm{LP}}$ is the Lie--Poisson form on 
$\mathfrak g^\ast$ and $\Omega_{\mathrm{can}}$ is the canonical form on 
$T^\ast(\mathcal{A}^s\times\mathcal{V}^s)$.

For variations $(\delta u,\delta \mu)$ in the Lie--Poisson sector one has
\[
\Omega_{\mathrm{LP}}\!\left((\dot u,\dot\mu),(\delta u,\delta\mu)\right)
 =
\langle \dot\mu,\delta u\rangle
 -\langle \delta\mu,\dot u\rangle
 +\langle \mu,[\dot u,\delta u]\rangle.
\]
The differential of the reduced Hamiltonian is
\begin{align*}
\dd H_{\mathrm{red}}
  &= \;\langle \mu , \delta u\rangle
   \;+\; \langle u , \delta\mu\rangle
   \;+\; \langle \alpha , (\delta A , \delta\Phi)\rangle
   \;+\; \langle (A,\Phi) , \delta\alpha\rangle \\[0.5em]
  &\quad
   -\left\langle \frac{\delta l}{\delta u}, \delta u \right\rangle
   -\left\langle \frac{\delta l}{\delta F}, \delta F \right\rangle
   -\left\langle \frac{\delta l}{\delta A}, \delta A \right\rangle
   -\left\langle \frac{\delta l}{\delta \Phi}, \delta \Phi \right\rangle .
\end{align*}

Matching the $\delta u$ and $\delta\mu$ terms in
\(
\iota_X \Omega_{\mathrm{red}} = \dd H_{\mathrm{red}}
\)
gives the system
\[
\begin{cases}
\dot\mu + \operatorname{ad}^\ast_u \mu
   = -\,F\diamond \dfrac{\delta l}{\delta F}
     \;+\;\mathcal{F}_{\mathrm{EM}},\\[0.9em]
\mu = \dfrac{\delta l}{\delta u}.
\end{cases}
\]
Substituting the second equation into the first yields the 
Euler--Poincaré equation
\[
\displaystyle
\frac{d}{dt}\!\left(\frac{\delta l}{\delta u}\right)
+\operatorname{ad}^\ast_u\left(\frac{\delta l}{\delta u}\right)
 =
F\diamond\left(\frac{\delta l}{\delta F}\right)
 +\mathcal{F}_{\mathrm{EM}}.
\]

The electromagnetic forcing \(\mathcal{F}_{\mathrm{EM}}\) comes from the 
variation of $l$ with respect to electromagnetic variables and is the 
phase--space Lorentz force density:
\[
\mathcal{F}_{\mathrm{EM}}(x,v)
   = 
   F(x,v)\Big(q\,E(x) + q\, v\times B(x)\Big),
\]
using the definitions
\[
E = -\partial_t A - \nabla\Phi,
\qquad
B = \nabla\times A.
\]

(3) The only terms involving $\delta F$ in the presymplectic identity 
come from the Lie--Poisson contribution and from the functional 
derivative $\delta l/\delta F$. Collecting the coefficients of 
$\delta F$ yields
\[
\big\langle \dot F + \mathcal L_u F ,\, \delta F \big\rangle = 0
\qquad\text{for all variations }\delta F.
\]
Therefore,
\[
\partial_t F + \mathcal L_u F = 0.
\]
This is the Eulerian form of the Vlasov equation, expressing that $F$ is 
advected by the Eulerian phase--space velocity field $u$.

(4) On the electromagnetic sector, the reduced presymplectic form is 
canonical:
\[
\Omega_{\mathrm{can}}
  = \dd A\wedge \dd P_A + \dd \Phi \wedge \dd P_\Phi.
\]
Thus the reduced SR equation gives the canonical Maxwell Hamilton 
equations:
\[
\dot A = \frac{\delta H_{\mathrm{red}}}{\delta P_A},\qquad
\dot P_A = -\frac{\delta H_{\mathrm{red}}}{\delta A},
\]
and similarly for $(\Phi,P_\Phi)$.

Imposing the GNH constraints $P_\Phi=0$ and 
$P_A = -E = -(\dot A+\nabla\Phi)$ yields
\[
\partial_t B = -\nabla\times E,
\qquad
\partial_t E = \nabla\times B - j(F,u),
\]
where the charge and current densities are
\[
\rho(F)(x)=\int F(x,v)\, dv,
\qquad
j(F,u)(x)=\int u_x(x,v)\, F(x,v)\, dv.
\]
The spatial velocity is $u_x(x,v)=v$, the $x$-projection of the 
Eulerian phase--space vector field.

The constraints $\nabla\cdot B=0$ and $\nabla\cdot E=\rho(F)$ are 
preserved, since they arise as compatibility conditions in the GNH 
procedure.

Combining (1)--(4) establishes that the reduced Skinner--Rusk dynamics is 
precisely the Euler--Poincaré--Maxwell formulation of the 
Maxwell--Vlasov system.
\end{proof}

\subsection{Hamiltonian Picture and Lie--Poisson Brackets}
\label{subsec:Hamiltonian-LP}

In this final step we show that the reduced presymplectic structure
$(W_{\mathrm{red}},\Omega_{\mathrm{red}})$ obtained from the 
Skinner--Rusk reduction reproduces, in a unified geometric fashion, the 
classical Hamiltonian structures for the Maxwell--Vlasov system, that is, the Marsden--Weinstein Lie--Poisson formulation 
      \cite{MarsdenWeinstein1982} and the the noncanonical Morrison--Greene bracket 
      \cite{MorrisonGreene1980,Morrison1982,Morrison1998}.

Recall from Proposition~\ref{prop:Omega-red} that $\Omega_{\mathrm{red}}
  = \Omega_{\mathrm{can}}
  +
  \Omega_{\mathrm{LP}}$, where $\Omega_{\mathrm{can}}$ is the canonical symplectic form on 
$T^\ast(\mathcal{A}^s\times\mathcal{V}^s)$ and $\Omega_{\mathrm{LP}}$ is the 
Lie--Poisson two--form associated with the semidirect-product group
\[
G_{\mathrm{sdp}}
  = \Diff^s(T\mathbb{R}^3)\ltimes
    \mathrm{Dens}^s(T\mathbb{R}^3).
\]
For any sufficiently smooth functional 
$\mathcal{F}: W_{\mathrm{red}}\to\mathbb{R}$, its Hamiltonian vector field 
$X_{\mathcal{F}}$ is defined by
\[
\iota_{X_{\mathcal F}} \Omega_{\mathrm{red}}
  = \dd\mathcal F,
\]
and the induced Poisson bracket on observables is
\[
\{\mathcal F,\mathcal G\}_{\mathrm{SR}}
   := \Omega_{\mathrm{red}}(X_{\mathcal F},X_{\mathcal G}).
\]

Restricting to functionals depending only on the Eulerian variables 
$z=(f,E,B)$, one obtains the Lie--Poisson bracket on the dual of the 
Lie algebra
\[
\mathfrak g_{\mathrm{sdp}}
  = \mathfrak X(T\mathbb{R}^3)\ltimes
    \mathrm{Dens}(T\mathbb{R}^3),
\]
together with the Maxwell part coming from $\Omega_{\mathrm{can}}$.  
Writing $\nabla\mathcal{F}$ for the variational derivatives, the bracket 
takes the standard form
\begin{equation}\label{eq:VM-bracket-complete}
\begin{aligned}
\{\mathcal F,\mathcal G\}(f,E,B)
&=
\underbrace{
\int_{\mathbb R^3}\!\int_{\mathbb R^3}
  f\,
  \bigg\{
    \frac{\delta \mathcal F}{\delta f},
    \frac{\delta \mathcal G}{\delta f}
  \bigg\}_{x,v}
  \,dx\,dv
}_{\text{(Vlasov canonical term)}} +\underbrace{
\int_{\mathbb R^3}\!\int_{\mathbb R^3}
  f\,
  B(x)\cdot
  \Big(
    \nabla_v \frac{\delta \mathcal F}{\delta f}
    \times
    \nabla_v \frac{\delta \mathcal G}{\delta f}
  \Big)\,dx\,dv
}_{\text{(magnetic twisting term)}}
\\[0.6em]
&\quad
+\underbrace{
\int_{\mathbb R^3}
\left(
  \frac{\delta \mathcal F}{\delta E}
    \cdot \big(\nabla\times \frac{\delta \mathcal G}{\delta B}\big)
 -
  \frac{\delta \mathcal G}{\delta E}
    \cdot \big(\nabla\times \frac{\delta \mathcal F}{\delta B}\big)
\right)dx
}_{\text{(pure Maxwell terms)}}
\\[0.6em]
&\quad
+\underbrace{
\int_{\mathbb R^3}\!\int_{\mathbb R^3}
  f\,
  \left(
    \frac{\delta\mathcal G}{\delta E}\cdot
      \nabla_v\frac{\delta\mathcal F}{\delta f}
    - 
    \frac{\delta\mathcal F}{\delta E}\cdot
      \nabla_v\frac{\delta\mathcal G}{\delta f}
  \right)
  dx\,dv
}_{\text{(field–particle coupling term)}}
\end{aligned}
\end{equation} which coincides exactly with the Lie--Poisson bracket derived in 
\cite{MarsdenWeinstein1982} for the Maxwell--Vlasov system. Here we are use that $
\{a,b\}_{x,v}
  = \nabla_x a\cdot\nabla_v b
    - \nabla_v a\cdot\nabla_x b$ is the canonical Poisson bracket on $T\mathbb{R}^3$.

\begin{corollary}
\label{cor:LiePoisson}
The reduced system 
$(W_{\mathrm{red}},\Omega_{\mathrm{red}},H_{\mathrm{red}})$ induces on the 
space of observables $\mathcal F(f,E,B)$ the Maxwell--Vlasov 
Lie--Poisson bracket \eqref{eq:VM-bracket-complete} where $f$ denotes the Eulerian distribution function.
The Hamiltonian flow generated by $X_{H_{\mathrm{red}}}$
satisfying $\iota_{X_{H_{\mathrm{red}}}}\Omega_{\mathrm{red}}
  = \dd H_{\mathrm{red}}$ coincides with the classical 
Marsden--Weinstein Hamiltonian formulation of the Maxwell--Vlasov system.

\end{corollary}

\medskip

We now make precise the relation with the Morrison--Greene 
noncanonical bracket.

\begin{theorem}
\label{thm:MW-Morrison}
Let $C_\infty\subset W_{\mathrm{red}}$ be the final constraint manifold
obtained by the GNH algorithm, and let
$\{\cdot,\cdot\}_{\mathrm{SR}}$ be the Poisson bracket induced by 
$\Omega_{\mathrm{red}}$ on observables of $(f,E,B)$. Then:

\begin{enumerate}

\item[(1)]
The quotient 
\[
C_\infty/\ker(\Omega_{\mathrm{red}}|_{C_\infty})
\]
is Poisson-isomorphic to the dual 
$\mathfrak g_{\mathrm{sdp}}^\ast$ of the semidirect-product Lie algebra 
used in the Marsden--Weinstein formulation 
\cite{MarsdenWeinstein1982}.  
Under this identification,
$\{\cdot,\cdot\}_{\mathrm{SR}}$ is exactly the Marsden--Weinstein 
Lie--Poisson bracket.

\item[(2)]
In Eulerian variables $(f,E,B)$, the same bracket reduces to the 
noncanonical Morrison--Greene bracket 
\cite{MorrisonGreene1980,Morrison1982,Morrison1998}:
\[
\{\mathcal F,\mathcal G\}_{\mathrm{MG}}
  =
  \int f\,
        \big\{\frac{\delta\mathcal F}{\delta f},
                 \frac{\delta\mathcal G}{\delta f}
        \big\}_{x,v}
      \,dx\,dv
  + \int\!
      \left(
        \frac{\delta\mathcal F}{\delta E}\cdot
        (\nabla\times \frac{\delta\mathcal G}{\delta B})
        -
        \frac{\delta\mathcal G}{\delta E}\cdot
        (\nabla\times \frac{\delta\mathcal F}{\delta B})
      \right)dx.
\]
Moreover, Morrison-type Casimir invariants (entropy of $f$, generalized 
helicity, Gauss invariants, etc.) correspond precisely to the 
functionals that are invariant under gauge transformations and 
particle-relabeling, i.e.\ constant along $\ker(\Omega_{\mathrm{red}})$.

\end{enumerate}
\end{theorem}

\begin{proof}
We divide the proof into four steps, reflecting the geometric passage
\[
(W,\Omega,H)
\;\longrightarrow\;
(W_{\mathrm{red}},\Omega_{\mathrm{red}},H_{\mathrm{red}})
\;\longrightarrow\;
C_\infty
\;\longrightarrow\;
C_\infty/\ker\Omega_{\mathrm{red}}
\;\cong\;
\mathfrak g_{\mathrm{sdp}}^\ast.
\]

\paragraph{\textbf{Step 1: Skinner--Rusk reduction and the semidirect product structure.}}

The space $W=TQ\oplus T^*Q$ carries the presymplectic form
$\Omega=\mathrm{pr}_2^\ast\omega_{\mathrm{can}}$.  
The particle--relabeling group $G=\Diff^s(T\mathbb R^3)$ acts freely and properly on $Q$ by right composition on the $\psi$-slot
and trivially on the electromagnetic potentials; its tangent and cotangent
lifts give a smooth presymplectic action on $W$.  
The action has a natural momentum map $J:T^*Q\longrightarrow \mathfrak g^{\,*}$, $\langle J(q,p),\xi\rangle
  = \langle p,\xi_Q(q)\rangle,$ and the Low Lagrangian is $G$-invariant.  
Thus the triple $(W,\Omega,H)$ is $G$-equivariant in the
sense of presymplectic reduction.

Reducing at any momentum value (in particular $\mu=0$) produces the weak
presymplectic manifold
\[
W_{\mathrm{red}}
  = J^{-1}(\mu)/G,
\qquad
\Omega_{\mathrm{red}}=\text{the basic form characterized by }
\pi^\ast\Omega_{\mathrm{red}}=\Omega|_{J^{-1}(\mu)}.
\]
As shown in Proposition~5.7, $\Omega_{\mathrm{red}}$ splits as $\Omega_{\mathrm{red}}
  = \Omega_{\mathrm{LP}} + \Omega_{\mathrm{can}}$,
where $\Omega_{\mathrm{LP}}$ is the infinitesimal coadjoint form for the
semidirect-product Lie algebra $\mathfrak g_{\mathrm{sdp}}
=
\mathfrak X(T\mathbb R^3)
\ltimes
\mathrm{Dens}(T\mathbb R^3)$, and $\Omega_{\mathrm{can}}$ is the canonical Maxwell symplectic form on
$T^*(\mathcal A^s\times\mathcal V^s)$.
Thus the geometry of $W_{\mathrm{red}}$ already carries the full
semidirect-product Lie--Poisson structure of
\cite{MarsdenWeinstein1982}.

\paragraph{\textbf{Step 2: GNH algorithm and the constraint manifold $C_\infty$.}}

Because $\Omega_{\mathrm{red}}$ is presymplectic, the equation $\iota_X\Omega_{\mathrm{red}}=\dd H_{\mathrm{red}}$ is not solvable on all of $W_{\mathrm{red}}$.  
Running the Gotay--Nester--Hinds algorithm produces a maximal submanifold
$C_\infty\subset W_{\mathrm{red}}$ where the equation is solvable and
where all constraint functions are preserved by the flow.

Geometrically,
$C_\infty$ removes:  
(i) gauge degeneracies in the Maxwell variables  
(ii) particle--relabeling degeneracies of the advected density  
(iii) integrability and compatibility constraints (Gauss law, 
$\nabla\cdot B=0$, etc.). Crucially, $\ker(\Omega_{\mathrm{red}}|_{C_\infty})$ contains precisely the directions generated by gauge transformations
$(A,\Phi)\mapsto(A+\nabla\chi,\Phi-\partial_t\chi)$ and by
particle-relabeling symmetries of the density $F$.
Thus quotienting by this kernel removes all unphysical directions.

\paragraph{\textbf{Step 3: Identification with the Marsden--Weinstein Lie--Poisson structure.}}
Consider the quotient $\mathcal P
  = C_\infty \big/ 
    \ker\!\left(\Omega_{\mathrm{red}}|_{C_\infty}\right)$. By general presymplectic reduction 
\cite{MarsdenWeinstein1974,Marle1985,OrtegaRatiu2004},  
$\mathcal P$ inherits a genuine Poisson structure defined by
\[
\{\mathcal F,\mathcal G\}_{\mathrm{SR}}
  = \Omega_{\mathrm{red}}(X_{\mathcal F},X_{\mathcal G}).
\]

Because the kernel corresponds exactly to the isotropy of the semidirect
action, the quotient is naturally isomorphic to the dual of the 
semidirect-product Lie algebra, $\mathcal P
\;\cong\;
\mathfrak g_{\mathrm{sdp}}^{\,*}$. Using the intrinsic pairing 
$\langle\,\cdot\,,\cdot\,\rangle:
   \mathfrak g_{\mathrm{sdp}}^{\,*}\times\mathfrak g_{\mathrm{sdp}}
   \to\mathbb{R}$,
a direct computation shows that the bracket decomposes as
\[
\{\mathcal F,\mathcal G\}_{\mathrm{SR}}
=
\left\langle 
    z,\,
    \big[
      \nabla\mathcal F(z)\!,
      \nabla\mathcal G(z)
    \big]
\right\rangle
\;+\;
\{\mathcal F,\mathcal G\}_{\mathrm{Max}}
\]
where the second term is the \emph{canonical Maxwell bracket}
arising from the electromagnetic canonical two–form:
\[
\{\mathcal F,\mathcal G\}_{\mathrm{Max}}
  =
  \int_{\mathbb{R}^3}
      \bigg(
          \frac{\delta\mathcal F}{\delta E}\cdot 
            \nabla\times\frac{\delta\mathcal G}{\delta B}
          \;-\;
          \frac{\delta\mathcal G}{\delta E}\cdot 
            \nabla\times\frac{\delta\mathcal F}{\delta B}
      \bigg)\,dx.
\]

Combining the Lie–Poisson part generated by the semidirect product
with the canonical Maxwell contribution yields precisely the
Marsden–Weinstein Maxwell–Vlasov bracket 
\cite{MarsdenWeinstein1982,MorrisonGreene1980,Morrison1982}.
This establishes item~(1).

\paragraph{\textbf{Step 4: Passage to Eulerian coordinates and the Morrison--Greene bracket.}}

Let \( z = (f,E,B) \) denote the Eulerian (physical) variables.  
For any observable \( \mathcal F = \mathcal F(f,E,B) \), the chain rule for 
variational derivatives gives
\[
\delta\mathcal F
  = \int\!\frac{\delta\mathcal F}{\delta f}(x,v)\,\delta f(x,v)\,dx\,dv
    \;+\;
    \int\!\frac{\delta\mathcal F}{\delta E}(x)\cdot\delta E(x)\,dx
    \;+\;
    \int\!\frac{\delta\mathcal F}{\delta B}(x)\cdot\delta B(x)\,dx .
\]

Substituting these variations into the Lie--Poisson bracket obtained in 
Step~3 (i.e.\ the bracket induced by \(\Omega_{\mathrm{LP}}+\Omega_{\mathrm{can}}\)), 
one computes each contribution explicitly:

\medskip
\noindent\emph{(1) Vlasov (particle) contribution.}
The Lie--Poisson part associated with the semidirect product
\(
\mathfrak g_{\mathrm{sdp}}
  = \mathfrak X(T\mathbb{R}^3)\ltimes\mathrm{Dens}(T\mathbb{R}^3)
\)
yields the Vlasov term
\[
\int f(x,v)\,
   \Big\{
     \frac{\delta\mathcal F}{\delta f},
     \frac{\delta\mathcal G}{\delta f}
   \Big\}_{x,v}\,dx\,dv,
\]
where $\{h,k\}_{x,v}
  = \nabla_x h\cdot \nabla_v k
    \;-\;
    \nabla_v h\cdot \nabla_x k$ is the canonical Poisson bracket on \(T\mathbb{R}^3\).

\medskip
\noindent\emph{(2) Field--particle cross terms.}
The mixed terms coupling \(f\) with \((E,B)\) arise from the 
semidirect-product action of phase--space diffeomorphisms on the 
electromagnetic potentials.  
These produce the standard Morrison field--particle contributions, which 
combine with the Vlasov part to encode the Lorentz-force coupling.

\medskip
\noindent\emph{(3) Pure Maxwell contribution.}
The electromagnetic part of the presymplectic form, 
\(\Omega_{\mathrm{can}}\), yields the canonical Maxwell bracket:
\[
\int\!\left(
    \frac{\delta\mathcal F}{\delta E}\cdot
        \big(\nabla\times\frac{\delta\mathcal G}{\delta B}\big)
    \;-\;
    \frac{\delta\mathcal G}{\delta E}\cdot
        \big(\nabla\times\frac{\delta\mathcal F}{\delta B}\big)
\right)\,dx.
\]

Combining (1)--(3) reproduces exactly the full 
Morrison--Greene bracket for the Vlasov--Maxwell system, 
showing item~(2).

Finally, a functional $\mathcal C$ is a Casimir iff
$X_{\mathcal C}\in\ker(\Omega_{\mathrm{red}})$, i.e., iff it is invariant
under gauge and particle-relabeling directions.  
These are precisely entropy-like functionals $\int \Phi(f)\,dx\,dv$; magnetic helicity and its generalizations; and Gauss constraints $\nabla\cdot B=0$ and $\nabla\cdot E=\rho(f)$. This matches exactly the Morrison–Greene Casimirs.
\end{proof}

\begin{remark}
The analysis above shows that the Skinner--Rusk formulation does not
merely coexist with the classical Hamiltonian descriptions of the
Maxwell--Vlasov system: it \emph{organizes and unifies} them within a
single geometric framework.  
Both the Marsden--Weinstein Lie--Poisson structure and the
Morrison--Greene noncanonical bracket arise as natural quotients of the
same underlying presymplectic object $(W_{\mathrm{red}},\Omega_{\mathrm{red}})$:
the former is obtained by passing to the coadjoint variables of the
semidirect-product group, while the latter results from expressing the
same bracket in Eulerian phase--space coordinates.

In this perspective, the degeneracies of $\Omega_{\mathrm{red}}$ encode
precisely the gauge and particle--relabeling symmetries, and the Casimir
invariants of the Morrison bracket correspond to functions that are
constant along the null leaves of $\Omega_{\mathrm{red}}$.  
Consequently, the familiar energy--Casimir method of plasma stability
acquires a clean geometric interpretation: it amounts to performing
constrained variations on $C_\infty$ along directions transverse to
$\ker(\Omega_{\mathrm{red}})$.

This unified viewpoint clarifies why seemingly different Hamiltonian
formulations of Maxwell--Vlasov share the same invariants and the same
stability theory: they are all shadows of a single presymplectic
structure before quotienting.\hfill$\diamond$
\end{remark}

\section{Symmetry Breaking and Reduced Dynamics Near Equilibria}

In this section we analyse how the Skinner--Rusk (SR) and 
Gotay--Nester--Hinds (GNH) formalisms adapt to the study of dynamics in 
a neighbourhood of a given Maxwell--Vlasov equilibrium.  
A distinctive feature of many physically relevant equilibria is the 
presence of \emph{partial symmetry breaking}: the full 
particle--relabeling symmetry group $G = \Diff^s(T\mathbb{R}^3)$ is no longer an exact symmetry of the equilibrium, which is fixed instead  
by a proper isotropy subgroup.  This leads naturally to a refined 
reduction scheme based on the residual symmetry and on a translated SR 
bundle around the equilibrium.



Let $(F_0,E_0,B_0)$ be a stationary solution of the Maxwell--Vlasov 
system in Eulerian variables, and let $(\psi_0,A_0,\Phi_0) \in Q$ be one of its Lagrangian representatives in the configuration space
$
Q = \Diff^s(T\mathbb{R}^3)\times \mathcal V^s\times \mathcal A^s$. We denote by $G = \Diff^s(T\mathbb{R}^3)$ the full particle--relabeling group, 
acting on $Q$ by
\[
g\cdot(\psi,\Phi,A) = (\psi\circ g^{-1},\,\Phi,\,A).
\] The \emph{isotropy subgroup} of the equilibrium 
representative $(\psi_0,A_0,\Phi_0)$ is
\[
G_0 := \Iso(\psi_0,A_0,\Phi_0)
  = \{\,g\in G \mid g\cdot(\psi_0,A_0,\Phi_0) = (\psi_0,A_0,\Phi_0)\,\}.
\]
Equivalently, on Eulerian variables $z_0 = (F_0,E_0,B_0)$,
\[
G_0 = \{\, g\in G \mid g\cdot z_0 = z_0 \,\}.
\]

\begin{definition}
An equilibrium $(F_0,E_0,B_0)$ is said to \emph{partially break} the 
$G$--symmetry if its isotropy subgroup $G_0$ is a proper subgroup of $G$, that is, $G_0 \subsetneq G$. We call $(G,G_0)$ the \emph{symmetry breaking pair} associated with the 
equilibrium.
\end{definition}

We now rewrite the unified Skinner--Rusk system in variables adapted to 
the equilibrium.  Let $(F,E,B) = (F_0,E_0,B_0) + (\widetilde F,\widetilde E,\widetilde B)$ and, in the configuration space, $(\psi,A,\Phi)
  = (\psi_0,A_0,\Phi_0) + (\widetilde\psi,\widetilde A,\widetilde\Phi)$ in an appropriate affine sense (e.g.\ via exponential charts). 
This induces a decomposition of the bundle $W = TQ \oplus T^*Q$ into the equilibrium point $w_0 = (q_0,v_0,p_0)$ and perturbations 
$(\widetilde q,\widetilde v,\widetilde p)$ around it.

\begin{definition}
The \emph{translated $W$ bundle} near the equilibrium $w_0$ is
\[
W_{F_0} := \{\, w - w_0 \mid w\in W \,\}
  \cong T_{w_0}W ,
\]
equipped with the pullback presymplectic form
\[
\Omega_{F_0} := \bigl(\tau_{w_0}\bigr)^*\Omega,
\]
where $\tau_{w_0}$ denotes translation by $w_0$.
The translated Hamiltonian is defined by
\[
H_{F_0}(\widetilde w)
  := H(w_0 + \widetilde w) - H(w_0).
\]
\end{definition}

\begin{proposition}
\label{prop:H-expansion}
Let $w_0\in W$ be an equilibrium of the SR--Maxwell--Vlasov system, and 
let $H_{F_0}$ be the translated Hamiltonian
\[
H_{F_0}(\widetilde w)
  := H(w_0+\widetilde w)-H(w_0).
\]
Then $H_{F_0}$ admits a Taylor expansion of the form
\[
H_{F_0}(\widetilde w)
  = H^{(2)}(\widetilde w)
    + H^{(3)}(\widetilde w)
    + \cdots,
\]
where $H^{(2)}$ is a quadratic form on $T_{w_0}W$ and 
$H^{(k)}$ is homogeneous of degree $k$ in $\widetilde w$ for $k\ge 3$.
Moreover, the isotropy subgroup $G_0$ of $w_0$ acts on the translated 
system $(W_{F_0},\Omega_{F_0},H_{F_0})$ by presymplectic automorphisms.
\end{proposition}

\begin{proof}
Since $w_0$ is an equilibrium, one has $\dd H(w_0)=0$.  
Hence the Taylor expansion of $H$ at $w_0$ reads
\[
H(w_0+\widetilde w)
  = H(w_0)
    + \frac12 \dd^{\,2}H(w_0)[\widetilde w,\widetilde w]
    + \sum_{k\ge 3} \frac{1}{k!}\dd^{\,k}H(w_0)[\widetilde w^{\otimes k}].
\]
Subtracting the constant value $H(w_0)$ yields the stated formula, with
\[
H^{(2)}(\widetilde w)
   := \frac12\,\dd^{\,2}H(w_0)[\widetilde w,\widetilde w],
\qquad
H^{(k)}(\widetilde w)
   := \frac{1}{k!}\dd^{\,k}H(w_0)[\widetilde w^{\otimes k}],\ k\ge3.
\]

By construction (see Lemma~\ref{prop:SR-invariance}), the action of $G$ on 
$W$ preserves $\Omega$, that is, $(g^\ast\Omega)(w)=\Omega(w)$ for all $g\in G$.

Since the translation $\tau_{w_0} : \widetilde w\mapsto w_0+\widetilde w$ 
is $G_0$–equivariant when restricted to the isotropy subgroup
($g\cdot w_0=w_0$ for $g\in G_0$), we have
\[
(g^\ast \Omega_{F_0})(\widetilde w)
  = (g^\ast \tau_{w_0}^\ast \Omega)(\widetilde w)
  = (\tau_{w_0}^\ast g^\ast \Omega)(\widetilde w)
  = \tau_{w_0}^\ast \Omega(\widetilde w)
  = \Omega_{F_0}(\widetilde w).
\]
Thus $\Omega_{F_0}$ is $G_0$–invariant.

Since $H$ is $G$--invariant and $g\cdot w_0=w_0$ for $g\in G_0$,
\[
H_{F_0}(g\cdot \widetilde w)
  = H(g\cdot(w_0+\widetilde w)) - H(w_0)
  = H(w_0+g\cdot\widetilde w) - H(w_0)
  = H(w_0+\widetilde w) - H(w_0)
  = H_{F_0}(\widetilde w).
\]

Because both $\Omega_{F_0}$ and $H_{F_0}$ are $G_0$--invariant, if 
$X_{H_{F_0}}$ is a solution of the presymplectic equation
\[
\iota_{X_{H_{F_0}}}\Omega_{F_0} = \dd H_{F_0},
\]
then for every $g\in G_0$ the pushforward $g_\ast X_{H_{F_0}}$ is also a 
solution:
\[
\iota_{g_\ast X_{H_{F_0}}}\Omega_{F_0}
  = \dd H_{F_0}.
\]
Thus $G_0$ acts on $(W_{F_0},\Omega_{F_0},H_{F_0})$ by presymplectic 
automorphisms, i.e.\ it maps solutions of the presymplectic dynamics to 
solutions.
\end{proof}

The first nontrivial contribution to the dynamics comes from the 
quadratic part $H^{(2)}$. The corresponding vector field $X^{(1)}$ on 
$W_{F_0}$ is defined by
\[
\iota_{X^{(1)}}\Omega_{F_0} = \dd H^{(2)} .
\]


We now make explicit the relation between the linear presymplectic 
dynamics and the usual linearized Maxwell--Vlasov equations around 
$(F_0,E_0,B_0)$. For concreteness we work in the non-relativistic 
setting with Hamiltonian density
\[
h(x,v;A,\Phi)
  = \frac{|v|^2}{2}
    + q\,\Phi(x)
    - q\,v\cdot A(x),
\]
so that the single-particle characteristics are generated by the 
canonical Poisson bracket $\{\cdot,\cdot\}_{x,v}$ on $T\mathbb{R}^3$.

Let $(\widetilde f,\widetilde E,\widetilde B)$ be Eulerian perturbations
of $(F_0,E_0,B_0)$, and denote
\[
E_0 = -\partial_t A_0 - \nabla\Phi_0,
\qquad
B_0 = \nabla\times A_0,
\]
with analogous definitions for the perturbed fields 
$(\widetilde E,\widetilde B)$ associated with 
$(\widetilde A,\widetilde\Phi)$.

Next, we provide the presymplectic formulation of the linearized Maxwell--Vlasov system.

\begin{theorem}
\label{thm:linearized}
The presymplectic system $(W_{F_0},\Omega_{F_0},H^{(2)})$, together with 
the GNH algorithm, yields a final linear constraint subspace 
$C_\infty^{(1)}\subset W_{F_0}$ on which the induced dynamics of 
$X^{(1)}$ is equivalent to the linearized Maxwell--Vlasov equations
around $(F_0,E_0,B_0)$:
\begin{align}
\partial_t \widetilde f
  &\;+\; \big\{\widetilde f,h_0\big\}_{x,v}
  \;+\; \big\{F_0,\widetilde h\big\}_{x,v}
   \;=\; 0, \label{eq:lin-Vlasov}\\[0.5em]
\partial_t \widetilde B
  &\;+\; \nabla\times \widetilde E
   \;=\; 0, \label{eq:lin-Maxwell-B}\\[0.3em]
\partial_t \widetilde E
  &\;-\; \nabla\times \widetilde B
   \;=\; -\,\widetilde j, \label{eq:lin-Maxwell-E}
\end{align}
together with the linearized constraints
\begin{equation}\label{eq:lin-constraints}
\nabla\cdot \widetilde B = 0,
\qquad
\nabla\cdot \widetilde E = \widetilde\rho,
\end{equation}
where
\[
h_0(x,v) := h(x,v;A_0,\Phi_0),
\qquad
\widetilde h(x,v) := q\,\widetilde\Phi(x) - q\,v\cdot \widetilde A(x),
\]
and the perturbed charge and current densities are
\[
\widetilde\rho(x)
  = \int_{\mathbb{R}^3} \widetilde f(x,v)\,dv,
\qquad
\widetilde j(x)
  = \int_{\mathbb{R}^3} v\,\widetilde f(x,v)\,dv.
\]
\end{theorem}

\begin{proof}
The construction parallels the nonlinear SR/GNH analysis but restricted 
to the quadratic Hamiltonian $H^{(2)}$.  
By Proposition~\ref{prop:H-expansion}, $H^{(2)}$ is precisely the second 
variation of the reduced Hamiltonian $H_{\mathrm{red}}$ at the 
equilibrium; its functional derivatives reproduce the linearization of 
the Euler--Poincaré--Maxwell equations of Theorem~\ref{thm:EP-from-SR}.

In Eulerian variables, the variation of the Hamiltonian with respect to 
$f$ yields the linearized Vlasov equation.  
Writing 
\(
\delta H^{(2)}/\delta f = h_0 + \widetilde h
\)
and using the Lie--Poisson part of $\Omega_{F_0}$, one finds
\[
\partial_t \widetilde f
  + \{ \widetilde f,h_0\}_{x,v}
  + \{F_0,\widetilde h\}_{x,v} = 0,
\]
which is precisely~\eqref{eq:lin-Vlasov}.

On the electromagnetic sector, the canonical part of $\Omega_{F_0}$ and 
the quadratic Hamiltonian yield the linear Maxwell system
\eqref{eq:lin-Maxwell-B}--\eqref{eq:lin-Maxwell-E}, with sources given 
by the first moments of $\widetilde f$.  
The GNH algorithm applied to $(\Omega_{F_0},H^{(2)})$ produces the 
linearized Gauss constraints~\eqref{eq:lin-constraints} and ensures that 
they are preserved by the linear flow.  

Thus the induced dynamics of $X^{(1)}$ on $C_\infty^{(1)}$ coincides with 
the standard linearized Maxwell--Vlasov system around 
$(F_0,E_0,B_0)$.
\end{proof}

The presence of an isotropy subgroup $G_0$ at equilibrium induces a 
natural splitting of the Lie algebra
\[
\mathfrak g = \mathfrak g_0 \oplus \mathfrak g_0^\perp,
\]
where $\mathfrak g_0$ is the Lie algebra of $G_0$ and $\mathfrak g_0^\perp$ 
is a chosen complement.

\begin{proposition}[Goldstone-type directions]
\label{prop:Goldstone}
Let $z_0=(F_0,E_0,B_0)$ be an equilibrium with isotropy subgroup 
$G_0\subset G=\Diff(T\mathbb{R}^3)$ and Lie algebra 
$\mathfrak g_0\subset\mathfrak g$.  
Then the linearization of the translated Skinner--Rusk system satisfies:

\begin{itemize}
\item Directions generated by $\mathfrak g_0$ correspond to 
\emph{unbroken symmetry directions}: their infinitesimal action 
lies in the presymplectic kernel 
$\ker(\Omega_{F_0}|_{C_\infty^{(1)}})$ and they annihilate 
$\dd H^{(2)}$.  
Hence they produce neutral modes associated with gauge and 
relabeling invariances.

\item Directions generated by a complement 
$\mathfrak g_0^\perp$ correspond to \emph{broken symmetry directions}.  
Their infinitesimal action moves $z_0$ along a family of physically 
equivalent equilibria, so they are annihilated by $\dd H^{(2)}$ even 
though they need not lie in $\ker\Omega_{F_0}$.  
These yield Goldstone-type neutral modes associated with continuous 
families of nearby equilibria.
\end{itemize}

Thus the tangent space decomposes as
\[
T_{z_0}W_{F_0}
  \;\supset\;
  \mathfrak g_0 
  \;\oplus\;
  \mathfrak g_0^\perp,
\]
where the first summand produces presymplectic degeneracies and the 
second produces dynamical degeneracies, giving a geometric version of 
Goldstone's theorem.
\end{proposition}

\begin{proof}
The kernel of $\Omega_{F_0}$ at the linear level consists of all 
infinitesimal symmetry directions of the unreduced Skinner--Rusk 
system, including particle-relabeling and gauge transformations.  
The fundamental vector field associated with $\xi\in\mathfrak g$ acts as
\[
\xi_{W_{F_0}}(z_0)=\frac{d}{d\epsilon}\Big|_{\epsilon=0}
   \exp(\epsilon\xi)\cdot z_0 .
\]

If $\xi\in\mathfrak g_0$ then $\exp(\epsilon\xi)\cdot z_0=z_0$ for all 
$\epsilon$.  Hence
\[
\xi_{W_{F_0}}(z_0)=0,
\qquad
\iota_{\xi_{W_{F_0}}}\Omega_{F_0}=0,
\qquad
\dd H^{(2)}(z_0)\cdot\xi_{W_{F_0}}=0.
\]
Thus $\mathfrak g_0$ generates directions that lie simultaneously in the 
kernels of $\Omega_{F_0}$ and of $\dd H^{(2)}$.  
These correspond to the unbroken gauge/relabeling symmetries, producing 
neutral modes of purely geometric origin.

If instead $\eta\in\mathfrak g_0^\perp$, then $z_0$ is not fixed by the 
flow $\exp(\epsilon\eta)$: the curve 
$\epsilon\mapsto\exp(\epsilon\eta)\cdot z_0$ lies entirely within the 
set of equilibria (because symmetry orbits preserve the Hamiltonian).  
Therefore
\[
\dd H(z_0)\cdot \eta_{W_{F_0}}=0
\qquad\Rightarrow\qquad
\dd H^{(2)}(z_0)\cdot\eta_{W_{F_0}}=0.
\]
Thus $\eta_{W_{F_0}}$ belongs to $\ker\dd H^{(2)}$ but typically not to 
$\ker\Omega_{F_0}$, yielding dynamical degeneracies associated with 
continuous symmetry-broken families.  

This reproduces the geometric mechanism underlying Goldstone modes:  
broken symmetries generate directions where the quadratic Hamiltonian is 
degenerate, producing neutrally stable linearized motions.
\end{proof}

The full linearized dynamics is defined on the final constraint subspace $C_\infty^{(1)} \subset W_{F_0}$, obtained by applying the GNH algorithm to $(\Omega_{F_0},H^{(2)})$.  
On this subspace, the restriction of $\Omega_{F_0}$ may still be 
degenerate.

\begin{definition}
The \emph{effective linearized phase space} around the equilibrium is 
the quotient
\[
W^{\mathrm{eff}}_{F_0}
  := C_\infty^{(1)} \big/ 
     \Bigl(\ker \Omega_{F_0}\big|_{C_\infty^{(1)}}\Bigr),
\]
endowed with the induced symplectic form $\Omega^{\mathrm{eff}}_{F_0}$.
\end{definition}

\begin{theorem}[Effective Hamiltonian linearization]
\label{thm:effective-linear}
The induced vector field
\[
X^{(1)}_{\mathrm{eff}} \in \mathfrak X(W^{\mathrm{eff}}_{F_0})
\]
defined by
\[
\iota_{X^{(1)}_{\mathrm{eff}}}\Omega^{\mathrm{eff}}_{F_0} 
  = \dd H^{(2)}_{\mathrm{eff}}
\]
provides a Hamiltonian description of the physically relevant 
perturbations around $(F_0,E_0,B_0)$, free of redundancies due to gauge 
and relabeling symmetries. Here $H^{(2)}_{\mathrm{eff}}$ is the quadratic 
Hamiltonian induced by $H^{(2)}$ on $W^{\mathrm{eff}}_{F_0}$.
\end{theorem}

\begin{proof}
By construction, the quotient by 
$\ker(\Omega_{F_0}|_{C_\infty^{(1)}})$ removes precisely the directions 
in which the presymplectic form vanishes.  
On the quotient, $\Omega^{\mathrm{eff}}_{F_0}$ is non-degenerate, so each 
sufficiently smooth functional admits a unique Hamiltonian vector field.  

The dynamics of $X^{(1)}$ projects to $X^{(1)}_{\mathrm{eff}}$ because 
kernel directions correspond to gauge and symmetry transformations that 
do not change physical observables.  
Thus the projected system is a genuine symplectic Hamiltonian system 
encoding the physical linear perturbations.
\end{proof}



The construction of the effective linearized phase space
$W^{\mathrm{eff}}_{F_0}$ requires a regularity condition ensuring that
the presymplectic form $\Omega_{F_0}$ has constant rank when restricted
to the final constraint submanifold $C_\infty^{(1)}$.  This ensures that
the quotient by $\ker(\Omega_{F_0}|_{C_\infty^{(1)}})$ is a smooth
manifold and that the induced two–form is symplectic.  Analogous
conditions appear in the presymplectic reduction theory developed in
\cite{ColomboMartindeDiegoZuccalli2010,ColomboDeDiego2014}.

\begin{definition}
We say that the translated presymplectic system
$(W_{F_0},\Omega_{F_0},H_{F_0})$ is \emph{regular at the equilibrium}
$z_0$ if the following two conditions hold:
\begin{enumerate}
\item[(i)] The distribution $K=\ker\bigl(\Omega_{F_0}|_{C_\infty^{(1)}}\bigr)$ has constant rank in a neighbourhood of $z_0$.
\item[(ii)] The tangent bundle of the constraint manifold splits as $T C_\infty^{(1)}
 \;=\;
K \;\oplus\; \mathcal{H}$, for some smooth subbundle $\mathcal{H}$.
\end{enumerate}
\end{definition}

Condition~(i) ensures that the quotient
$C_\infty^{(1)}/K$ is a smooth manifold; condition~(ii) guarantees that
the pushforward of $\Omega_{F_0}$ along the quotient map is nondegenerate
on the image of $\mathcal{H}$.  Together these imply that the effective
phase space is a genuine symplectic manifold.

\begin{proposition}
\label{prop:symplectic-regularity}
If $(W_{F_0},\Omega_{F_0},H_{F_0})$ is regular at $z_0$ in the sense
above, then:
\begin{enumerate}
\item The effective phase space $W^{\mathrm{eff}}_{F_0}
  = C_\infty^{(1)} \big/ 
    \ker(\Omega_{F_0}|_{C_\infty^{(1)}})$ is a smooth symplectic manifold.

\item The quadratic Hamiltonian $H^{(2)}$ induces a smooth function
$H^{(2)}_{\mathrm{eff}}$ on $W^{\mathrm{eff}}_{F_0}$.

\item The effective linearized vector field $X^{(1)}_{\mathrm{eff}}$ of
Theorem~\ref{thm:effective-linear} is well-defined and smooth.
\end{enumerate}
\end{proposition}

\begin{proof}
Regularity condition (i) implies that the kernel distribution $K$ is an
integrable, constant-rank subbundle, so the quotient
$C_\infty^{(1)}/K$ is a smooth manifold.  
Condition (ii) implies that the restriction of $\Omega_{F_0}$ to the
horizontal subbundle $\mathcal{H}$ is nondegenerate.  
This form descends to the quotient, yielding the symplectic form
$\Omega^{\mathrm{eff}}_{F_0}$.
Compatibility with the Hamiltonian follows from the construction of
$H^{(2)}_{\mathrm{eff}}$ and the identity
\[
\iota_{X^{(1)}_{\mathrm{eff}}}
  \Omega^{\mathrm{eff}}_{F_0}
   = \dd H^{(2)}_{\mathrm{eff}}.
\]
\end{proof}

\begin{remark}
In applications, regularity holds for a large class of
Maxwell--Vlasov equilibria of physical relevance (e.g.\ isotropic,
axisymmetric, or spatially homogeneous equilibria).  
Failure of regularity corresponds to singular equilibria such as
magnetic nulls, degenerate current sheets, or equilibria lying on
singular strata of the coadjoint representation.
\hfill$\diamond$
\end{remark}

In the next section, we will use the effective phase space 
$W^{\mathrm{eff}}_{F_0}$ and the quadratic Hamiltonian 
$H^{(2)}_{\mathrm{eff}}$ as the natural setting for energy--Casimir 
stability analysis, relating the second variation of energy--Casimir 
functionals to the spectrum of the linearized Maxwell--Vlasov operator.

\section{Energy--Casimir Methods and Stability in the Reduced SR Framework}

In this section we show how the presymplectic reduction of the
Skinner--Rusk formalism provides a natural geometric setting for
energy--Casimir methods and stability analysis of Maxwell--Vlasov
equilibria. The key point is that, once the Weinstein--Morrison
Lie--Poisson bracket is recovered from the reduced data, the usual
energy--Casimir constructions can be interpreted as constrained
variations on the final presymplectic constraint manifold, modulo the
symmetries and degeneracies described in Sections~5.

Throughout, we fix a stationary solution $(F_0,E_0,B_0)$ with isotropy
subgroup $G_0\subset G = \Diff(T\mathbb{R}^3)$ and consider the translated
bundle $(W_{F_0},\Omega_{F_0},H_{F_0})$ and its linearized structure as
introduced in Section~5.


Let $\mathcal{P}$ denote the Poisson manifold carrying the
Maxwell–Vlasov Lie–Poisson bracket obtained in Section~\ref{subsec:Hamiltonian-LP}, with coordinates $(F,E,B)$.
Equilibria are commonly characterized as critical points of an
energy--Casimir functional
\[
\mathcal{E}_C(F,E,B)
  := H_{\mathrm{MV}}(F,E,B) + C(F,E,B),
\]
where $H_{\mathrm{MV}}$ is the Maxwell--Vlasov Hamiltonian and $C$ is a
Casimir functional.

\begin{definition}
A smooth functional $C$ on $\mathcal{P}$ is called a
\emph{Casimir of the Maxwell–Vlasov Lie–Poisson bracket} if it Poisson commutes with all
observables:
\[
\{C,\mathcal{F}\}_{\mathrm{MV}} = 0
\qquad \text{for all functionals }\mathcal{F}.
\]
\end{definition}

In the SR/GNH picture, Casimirs admit the following interpretation.

\begin{proposition}
\label{prop:Casimirs-null}
Let $C$ be a functional on $\mathcal{P}$ and let
$\Omega_{\mathrm{red}}$ be the reduced presymplectic form on
$C_\infty$ as in Theorem~\ref{thm:MW-Morrison}. Then $C$ is a
Casimir of the Maxwell–Vlasov Lie–Poisson bracket if and only if its pullback to $C_\infty$
is constant along $\ker\bigl(\Omega_{\mathrm{red}}|_{C_\infty}\bigr)$;
equivalently,
\[
\dd C(z)\cdot v = 0
\qquad
\text{for all }z\in C_\infty \text{ and all }v\in\ker \Omega_{\mathrm{red}}(z).
\]
\end{proposition}

\begin{proof}
The Hamiltonian vector field $X_C$ is defined on the quotient of
$C_\infty$ by $\ker\Omega_{\mathrm{red}}$, and the Poisson bracket is
given by $\{C,\mathcal{F}\} = \Omega_{\mathrm{red}}(X_C,X_{\mathcal{F}})$.
Thus $\{C,\mathcal{F}\}=0$ for all $\mathcal{F}$ if and only if
$X_C$ lies in the kernel of $\Omega_{\mathrm{red}}$, i.e. if and only if
$C$ is constant along presymplectic null directions.
\end{proof}

\begin{remark}
In concrete plasma models, Casimirs include the usual entropy-type
functionals $\int \Phi(F)\,dx\,dv$ and electromagnetic invariants.
Proposition~\ref{prop:Casimirs-null} shows that these are exactly the
functionals that are insensitive to particle-relabeling and gauge
transformations, in agreement with the geometric picture of
Sections~5.\hfill$\diamond$
\end{remark}


We now reinterpret Maxwell--Vlasov equilibria as constrained critical
points of energy--Casimir functionals on the presymplectic manifold
$C_\infty$.

\begin{theorem}
\label{thm:EC-critical}
Let $(F_0,E_0,B_0)$ be a stationary solution of the Maxwell--Vlasov
system, and let $\mathcal{E}_C = H_{\mathrm{MV}} + C$ be an
energy--Casimir functional. The following are equivalent:
\begin{enumerate}
\item $(F_0,E_0,B_0)$ is an equilibrium for the Maxwell--Vlasov flow.

\item The corresponding point $z_0\in C_\infty$ is a critical point of
$\mathcal{E}_C$ restricted to $C_\infty$:
\[
\dd \mathcal{E}_C(z_0)\cdot \delta z = 0
\qquad\text{for all } \delta z\in T_{z_0}C_\infty.
\]

\item The Hamiltonian vector field $X_{\mathcal{E}_C}$ vanishes at $z_0$
modulo $\ker\Omega_{\mathrm{red}}$, i.e. $X_{\mathcal{E}_C}(z_0)$ is a
pure gauge/relabeling direction.
\end{enumerate}
\end{theorem}

\begin{proof}
(1) $\Rightarrow$ (3) follows from the fact that equilibria are
exactly the fixed points (up to gauge) of the Hamiltonian flow
generated by $H_{\mathrm{WM}}$, and Casimirs do not contribute to the
Hamiltonian vector field. The equivalence between (2) and (3) is a
standard consequence of the presymplectic identity
$i_{X_{\mathcal{E}_C}}\Omega_{\mathrm{red}} = \dd \mathcal{E}_C$ on
$C_\infty$, together with the presence of null directions.
Finally, (2) $\Rightarrow$ (1) follows by projecting to the Poisson
manifold $\mathcal{P}$ and using the Maxwell–Vlasov Lie–Poisson bracket obtained in Section~\ref{subsec:Hamiltonian-LP}.
\end{proof}


We consider now the second variation of $\mathcal{E}_C$ at an
equilibrium and derive a formal stability criterion.

\begin{definition}
Let $z_0\in C_\infty$ be an equilibrium as in
Theorem~\ref{thm:EC-critical}. The \emph{second variation} of
$\mathcal{E}_C$ at $z_0$ is the symmetric bilinear form
\[
\delta^2\mathcal{E}_C(z_0)(\delta z_1,\delta z_2)
  := \left.\frac{\partial^2}{\partial s\,\partial t}\right|_{s=t=0}
      \mathcal{E}_C\bigl(z_0 + s\delta z_1 + t\delta z_2\bigr),
\]
defined for $\delta z_1,\delta z_2\in T_{z_0}C_\infty$.
\end{definition}

Because of gauge and relabeling symmetries, $\delta^2\mathcal{E}_C(z_0)$
is necessarily degenerate along $\ker\Omega_{\mathrm{red}}$.

\begin{theorem}
\label{thm:EC-stability}
Assume that the quadratic form $\delta^2\mathcal{E}_C(z_0) : T_{z_0}C_\infty \times T_{z_0}C_\infty
  \to \mathbb{R}$ is positive definite on a complement of
$\ker\Omega_{\mathrm{red}}(z_0)$ in $T_{z_0}C_\infty$, i.e. there is a
subspace $V\subset T_{z_0}C_\infty$ such that
\[
T_{z_0}C_\infty = V \oplus \ker\Omega_{\mathrm{red}}(z_0)
\quad\text{and}\quad
\delta^2\mathcal{E}_C(z_0)(\delta z,\delta z) > 0
\ \forall\,0\neq\delta z\in V.
\]
Then $(F_0,E_0,B_0)$ is \emph{formally stable} in the sense of
energy--Casimir: any perturbation that preserves the Casimirs and
remains on $C_\infty$ cannot decrease $\mathcal{E}_C$ to second order.
\end{theorem}

\begin{proof}
Since $z_0$ is an equilibrium of the reduced Maxwell--Vlasov flow, it is a
critical point of the constrained functional $\mathcal{E}_C$ on the
presymplectic manifold $(C_\infty,\Omega_{\mathrm{red}})$.  
Thus the first variation satisfies
\[
\dd\mathcal{E}_C(z_0)\cdot \delta z = 0
\qquad\text{for every }\delta z\in T_{z_0}C_\infty.
\]

Let $\mathcal{K}:=\ker\bigl(\Omega_{\mathrm{red}}|_{C_\infty}\bigr)
\subset T_{z_0}C_\infty$ be the characteristic distribution of the presymplectic form.
By construction (Section~5), $\mathcal{K}$ consists precisely of
infinitesimal gauge and relabeling transformations; variations along
$\mathcal{K}$ do not affect any gauge-invariant or observable quantity.

We choose a complementary subspace $T_{z_0}C_\infty = V \oplus \mathcal{K}$, where $V$ is transversal to the characteristic foliation.  This is the
standard decomposition for presymplectic stability analysis. For any admissible perturbation $\delta z$, write uniquely $\delta z = v + k,
\,
v\in V,\; k\in\mathcal{K}$. Because every Casimir $C$ satisfies
\[
\dd C(z)\cdot k = 0
\qquad\text{for all }k\in\mathcal{K},
\]
and because the energy functional is $G$-invariant under gauge and
relabeling symmetries, we have
\[
\dd\mathcal{E}_C(z_0)\cdot k = 0,
\qquad
\dd^2\mathcal{E}_C(z_0)[k,k'] = 0,
\]
for all $k,k'\in\mathcal{K}$.  
Thus the second variation in any direction containing a component in
$\mathcal{K}$ depends only on its projection to~$V$:
\[
\dd^2\mathcal{E}_C(z_0)[v+k,v+k]
   = \dd^2\mathcal{E}_C(z_0)[v,v].
\]

By hypothesis, the reduced Hessian satisfies
\[
\dd^2\mathcal{E}_C(z_0)[v,v] > 0
\qquad\text{for all nonzero }v\in V.
\]
Hence $\mathcal{E}_C$ attains a strict local minimum on $z_0$ modulo the
characteristic directions $\mathcal K$.

By the implicit function theorem on presymplectic manifolds (see
\cite{GotayNesterHinds1978}, \cite{MarsdenRatiu1999}), this implies that
for every sufficiently small perturbation the value of $\mathcal E_C$ is
bounded below by its value at $z_0$ up to a gauge or relabeling
transformation.  
Denote by $X_H$ the Hamiltonian vector field on $(C_\infty,\Omega_{\mathrm{red}})$
defined by $\iota_{X_H}\Omega_{\mathrm{red}}=\dd H_{\mathrm{red}}$.
Then
\[
\mathcal{E}_C(z_0) 
  \le \mathcal{E}_C(\exp(t X_H)(z))
  + o(\|z-z_0\|_{V})
\qquad
\text{for all $z$ near $z_0$}.
\]

Since the dynamics on $C_\infty$ preserves the presymplectic form and
hence the characteristic distribution $\mathcal K$, the flow
$\phi_t$ of the reduced Hamiltonian vector field $X_H$ satisfies $\phi_t(z) = g_t \cdot z(t)$, with $g_t$ a gauge/relabeling transformation and $z(t)$ evolving on the
quotient leaf space.  
Because $\mathcal{E}_C$ is constant on $\mathcal K$ and positive
definite on $V$, $\mathcal{E}_C$ is a Lyapunov function for the
equivalence class of $z_0$ modulo $\mathcal K$.

Thus the trajectory cannot drift away from the equilibrium leaf, and its
projection to the reduced symplectic space is bounded for all times for
which the solution exists smoothly.

This establishes formal energy--Casimir stability of $z_0$ in the sense
of \cite{Morrison1998,MorrisonGreene1980}.
\end{proof}

\begin{remark}
In many examples, the choice of Casimirs $C$ amounts to imposing
constraints such as fixed total charge, entropy, or magnetic flux.
Theorem~\ref{thm:EC-stability} states that if $\mathcal{E}_C$ has a
strict local minimum on $C_\infty$ modulo the degeneracies generated by
$G_0$ and gauge transformations, then the corresponding equilibrium is
formally stable.\hfill$\diamond$
\end{remark}


The linearized effective phase space $W^{\mathrm{eff}}_{F_0}$ of
Section~5 provides an appropriate arena for spectral analysis and for
relating the second variation to the spectrum of the linearized
operator.

\begin{proposition}
\label{prop:Heff-quadratic}
Let $H^{(2)}_{\mathrm{eff}}$ be the quadratic Hamiltonian induced by
$H^{(2)}$ on the effective symplectic space $W^{\mathrm{eff}}_{F_0}$.
Then:
\begin{enumerate}
\item $H^{(2)}_{\mathrm{eff}}$ coincides with the restriction of
$\delta^2\mathcal{E}_C(z_0)$ to $W^{\mathrm{eff}}_{F_0}$, up to
Casimir terms that vanish to second order.
\item The Hamiltonian flow generated by $H^{(2)}_{\mathrm{eff}}$ is the
projection of the linearized Maxwell--Vlasov dynamics onto the space
of physical perturbations.
\end{enumerate}
\end{proposition}

\begin{proof}
(1) follows from the definition of $H_{F_0}$ as a translated Hamiltonian
and from the fact that Casimir contributions do not alter the
Hamiltonian vector field. Restriction to $W^{\mathrm{eff}}_{F_0}$ simply
removes gauge and symmetry directions. (2) is immediate from the
construction of $W^{\mathrm{eff}}_{F_0}$ as the quotient of
$C_\infty^{(1)}$ by $\ker\Omega_{F_0}$.
\end{proof}

The combination of Theorems~\ref{thm:EC-stability} and
\ref{thm:effective-linear}, together with
Proposition~\ref{prop:Heff-quadratic}, shows that the SR/GNH
framework not only recovers the the classical Hamiltonian formulations of the Maxwell–Vlasov system, but also organizes the energy--Casimir stability analysis in
a presymplectic language that is directly adapted to symmetry
breaking, constraints and gauge invariances.

\section{Controlled Symmetry Breaking and Hamiltonian Control
in the Skinner--Rusk Maxwell--Vlasov Framework}

In this section we extend the presymplectic SR--GNH formulation of the
Maxwell--Vlasov system to include \emph{affine Hamiltonian controls}
representing external antenna fields or radio frequency (RF) perturbations.  
Such drivings, classically studied in the oscillation--center and
ponderomotive theories of Kaufman, Littlejohn, Cary--Brizard, and
Similon et al.~\cite{Kaufman1986,Littlejohn1983,CaryBrizard2009,Similon1985},
fit naturally into the geometric structure developed in the previous section.  The SR framework provides a unified, coordinate-free setting
for control-induced symmetry breaking near equilibria, presymplectic
Hamiltonian response to antenna forcing, and controlled stabilization
via affine Hamiltonian feedback.

The key observation is that the reduced system of
Sections~5--6 provides a natural presymplectic arena for formulating
\emph{controlled Lie--Poisson dynamics} on the particle--field phase
space, in which external fields can be used to act on Goldstone-type
neutral modes or to move an equilibrium onto a favorable Casimir leaf.



Let $(W_{\mathrm{red}},\Omega_{\mathrm{red}},H_{\mathrm{red}})$ denote
the reduced Skinner--Rusk Maxwell--Vlasov system, with reduced
coordinates $z = (f,E,B)$. We introduce an \emph{affine Hamiltonian control} term of the form
\begin{equation}\label{eq:affine-control}
H_u(z)
  := H_{\mathrm{red}}(z)
     + \sum_{a=1}^m u_a(t)\, B_a(z),
\end{equation}
where $u(t)=(u_1(t),\dots,u_m(t))\in U\simeq\mathbb{R}^m$ encodes the
antenna driving amplitudes, and each $B_a \colon W_{\mathrm{red}}\to\mathbb{R}$ is a smooth functional describing the coupling of the plasma to the
$a$th antenna mode.  

In a potential formulation of the electromagnetic sector, in which
$(A,\Phi)$ are retained as variables (subject to gauge constraints), a
typical example is
\[
B_a(z)
   = \int_{\mathbb{R}^3} J_{\mathrm{ext},a}(x)\cdot A(x)\,dx,
\]
where $J_{\mathrm{ext},a}$ is the externally prescribed current profile
associated with the $a$th antenna, and $A$ is the plasma vector
potential.  This term represents the standard current--potential
coupling.  More general couplings may also include electrostatic terms
$\int \rho_{\mathrm{ext},a}(x)\,\Phi(x)\,dx$ or RF wave envelopes in
oscillation--center variables.

The associated Hamiltonian vector field $X_{H_u}$ satisfies the
presymplectic equation
\[
\iota_{X_{H_u}}\Omega_{\mathrm{red}} = \dd H_u ,
\]
so the controlled evolution remains compatible with the SR--GNH
constraint structure and reproduces the Maxwell equations with external
sources in the usual way.

\begin{remark}
In many plasma--antenna models
(e.g.\ Kaufman~\cite{Kaufman1986}, Similon et al.~\cite{Similon1985}),
each control channel $u_a(t)$ represents the time-dependent amplitude
of a prescribed antenna mode (e.g.\ an RF drive).  The corresponding
Hamiltonian contribution $B_a(z)$ is then obtained by evaluating the
interaction energy of the antenna with the plasma fields: in a
potential description, this leads to current--potential couplings such
as
\[
B_a(z)
  = \int_{\mathbb{R}^3}
        J_{\mathrm{ext},a}(x)\cdot A(x)\,dx.
\]
In oscillation--center theory, $u_a(t)$ may be a rapidly oscillating
amplitude, and the functionals $B_a$ encode the wave--particle
interaction phase.\hfill$\diamond$
\end{remark}


Let $z_0=(f_0,E_0,B_0)$ be a Maxwell--Vlasov equilibrium with isotropy
subgroup $G_0\subset\Diff(T\mathbb{R}^3)$ and corresponding Lie algebra
$\mathfrak{g}_0$.  In the translated SR picture of Section~6, the
tangent space of the translated manifold at the corresponding point
admits a natural decomposition of the form
\[
T_{z_0}W_{F_0}
  \;\cong\;
  \mathfrak{g}_0
  \;\oplus\;
  \mathfrak{g}_0^\perp
  \;\oplus\;
  C_\infty^{(1)},
\]
where $\mathfrak{g}_0$ generates the unbroken particle--relabeling
symmetries, $\mathfrak{g}_0^\perp$ generates the broken symmetry
directions (Goldstone modes), and $C_\infty^{(1)}$ is the linear GNH
constraint subspace introduced in Section~6.  After quotienting by the
presymplectic kernel, the physical degrees of freedom lie in the
effective phase space
\[
W^{\mathrm{eff}}_{F_0} := 
C_\infty^{(1)}\big/
\ker\bigl(\Omega_{F_0}\big|_{C_\infty^{(1)}}\bigr),
\]
endowed with the induced symplectic form
$\Omega^{\mathrm{eff}}_{F_0}$.

Figure~\ref{fig:control-diagram} illustrates schematically how the
external antenna vector fields associated with $B_a$ act on the reduced
manifold: they deform the constraint hierarchy produced by the GNH
algorithm and modify the foliation of the reduced Poisson manifold
into Casimir leaves.

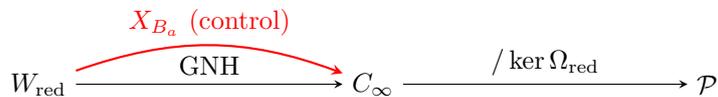
\begin{figure}[h!]
\centering
\begin{tikzpicture}[>=stealth,scale=1.1]
\node (W) at (0,0) {$W_{\mathrm{red}}$};
\node (C) at (4,0) {$C_\infty$};
\node (P) at (8,0) {$\mathcal{P}$};
\draw[->] (W) -- node[above] {$\text{GNH}$} (C);
\draw[->] (C) -- node[above] {$/\ker\Omega_{\mathrm{red}}$} (P);
\draw[->,red,thick,bend left=20] (W) to node[above]
{$X_{B_a}$ (control)} (C);
\end{tikzpicture}
\caption{Conceptual diagram: antenna control vector fields
$X_{B_a}$ act on the SR manifold before and after GNH reduction,
modifying symmetry directions and Casimir leaves.}
\label{fig:control-diagram}
\end{figure}

\begin{definition}
A control channel $u_a(t)$ is said to \emph{break the residual isotropy}
at $z_0$ if
\[
X_{B_a}(z_0)\notin T_{z_0}(G_0\cdot z_0),
\]
where $G_0\cdot z_0$ denotes the $G_0$-orbit of the equilibrium.
\end{definition}

Such controls ``lift'' neutral modes arising from broken symmetry
directions into controlled directions, allowing stabilization strategies
analogous to those for Hamiltonian systems with symmetry.



We now state the main result of this section, combining the
symmetry-breaking analysis of Section~5 with the energy--Casimir
machinery of Section~6 and the affine control structure introduced
above.

\begin{theorem}
\label{thm:controlled-sb}
Let $z_0$ be a Maxwell--Vlasov equilibrium with isotropy group $G_0$,
and let $W^{\mathrm{eff}}_{F_0}$ be the effective symplectic
linearized phase space around $z_0$.  Consider the controlled Hamiltonian
$H_u$ defined in~\eqref{eq:affine-control}.  Suppose that:

\begin{enumerate}
\item[(i)] For every broken symmetry direction
$Y\in\mathfrak{g}_0^\perp$ there exists a control component
$B_a$ such that the corresponding Hamiltonian vector field satisfies
\[
\Omega^{\mathrm{eff}}_{F_0}\bigl(X_{B_a}(z_0),Y\bigr)\neq 0.
\]

\item[(ii)] There exists a constant control $u^{*}\in U$ such that
\[
\dd\mathcal{E}_C(z_0) + 
\sum_{a=1}^m u^{*}_a\,\dd B_a(z_0) = 0
\quad\text{on } T_{z_0}C_\infty,
\]
i.e. $u^{*}$ shifts $z_0$ to a critical point of the controlled
energy--Casimir functional
\[
\mathcal{E}_{C,u}(z)
  := H_{\mathrm{red}}(z)+C(z)
      +\sum_a u^{*}_a B_a(z).
\]

\item[(iii)] The quadratic form
\[
\delta^2\mathcal{E}_{C,u}(z_0)
=
\delta^2\mathcal{E}_C(z_0)
\;+\;
\sum_{a=1}^m u^{*}_a\,
\delta^2 B_a(z_0)
\]
is positive definite on $W^{\mathrm{eff}}_{F_0}$.
\end{enumerate}

Then:

\begin{enumerate}
\item[(a)] $z_0$ is a nondegenerate critical point of $H_{u^{*}}$
on the constraint manifold $C_\infty$.

\item[(b)] $z_0$ is formally energy--Casimir stable under the controlled
Hamiltonian flow generated by $H_{u^{*}}$.

\item[(c)] The controlled linearized dynamics on
$W^{\mathrm{eff}}_{F_0}$ is symplectic and stable:
\[
i_{X^{(1)}_{u^{*}}}\,\Omega^{\mathrm{eff}}_{F_0}
=
\dd\Bigl(H^{(2)}_{\mathrm{eff}}
      +\sum_a u^{*}_a B_a\Bigr),
\]
where $H^{(2)}_{\mathrm{eff}}$ is the quadratic Hamiltonian induced by
$H_{\mathrm{red}}$ on $W^{\mathrm{eff}}_{F_0}$.
\end{enumerate}
\end{theorem}

\begin{proof}
The proof follows the structure of
Theorems~\ref{thm:effective-linear} and~\ref{thm:EC-stability}.
Condition~(i) guarantees that the control vector fields span the
presymplectic null directions associated with broken symmetries at
$z_0$, so that no uncontrolled neutral directions remain in
$W^{\mathrm{eff}}_{F_0}$.  Condition~(ii) ensures that $z_0$ is a
critical point of the controlled energy--Casimir functional
$\mathcal{E}_{C,u}$ when restricted to $C_\infty$.  Condition~(iii)
then implies that this critical point is a strict local minimum modulo
gauge and relabeling degeneracies.  The standard energy--Casimir
argument on the effective phase space then yields (b), while (a) and
(c) follow from the presymplectic reduction and the definition of
$W^{\mathrm{eff}}_{F_0}$.
\end{proof}

\subsection{Controlled Lie--Poisson Reduction for Plasma--Antenna Coupling}
\label{sec:controlled-LP}

In this subsection we develop a geometric theory of \emph{controlled
Lie--Poisson reduction} for the Maxwell--Vlasov system, interpreting an
external antenna as an \emph{affine cotangent-lift control} acting on
the Skinner--Rusk (SR) presymplectic structure.  
The construction produces a new class of \emph{controlled semidirect-product
Lie--Poisson equations} compatible with gauge symmetry,
particle–relabeling symmetry, and with the hierarchy of constraints
generated by the GNH algorithm.

Throughout, $G=\Diff(T\mathbb{R}^3)$ denotes the particle–relabeling
group acting on the phase space $W=TQ\oplus T^*Q$ with 
$Q=G\times\mathcal{A}\times\mathcal{V}$. Reduction of the uncontrolled SR system by $G$ yields the
Marsden--Weinstein Lie--Poisson structure on the semidirect-product dual
$\g^*_{\mathrm{sdp}}$.  
We now extend this picture to include external RF or antenna fields.


Let $U\subset\mathbb{R}^m$ be the control space.  
An antenna driving field is modeled as an affine deformation of the
cotangent–lift action:
\begin{equation}\label{eq:affine-cotangent-lift}
\mathcal{A}_u : T^*Q \to T^*Q,
\qquad
\mathcal{A}_u(\alpha_q) = \Phi^\ast(\alpha_q) + B(q)\,u,
\end{equation}
where $\Phi\in\Diff(Q)$ is induced by the uncontrolled dynamics, and $B : Q \rightarrow \operatorname{Lin}(\mathbb{R}^m, T_q^*Q)$ is a $G$-equivariant control one–form encoding the antenna–plasma
coupling. Here \(
\operatorname{Lin}(\mathbb{R}^m, T_q^*Q)
\)
denotes the vector space of linear maps from the control space 
\(\mathbb{R}^m\) into the cotangent space \(T_q^*Q\).  
Equivalently, each
\(
B(q)\in \operatorname{Lin}(\mathbb{R}^m, T_q^*Q)
\)
may be viewed as an $m$--tuple of covectors at $q$, or as a bundle map
\(
B \colon Q \to T^*Q \otimes (\mathbb{R}^m)^*.
\)
Thus, for a control input \(u(t)\in\mathbb{R}^m\), the composite
\(B(q)u(t)\in T_q^*Q\) represents the affine cotangent--lift contribution
to the controlled dynamics.

\begin{definition}
A \emph{controlled system} is a quadruple $(W,\Omega,H,B)$ consisting
of a presymplectic manifold $(W,\Omega)$, an energy function $H$, and a
$G$-equivariant control one-form $B$, such that the controlled dynamics
is defined by
\begin{equation}\label{eq:controlled-SR}
i_{X_{H,u}}\Omega
   = \dd H + B^\top(\,\cdot\,)\,u(t).
\end{equation}
\end{definition}

$G$-equivariance of $B$ implies
\[
g^\ast B = B
\qquad\text{and}\qquad
g_\ast X_{H,u} = X_{H,u}\circ g ,
\]
so the control preserves (or intentionally breaks) particle–relabeling
symmetry in a controlled manner.


The controlled SR equation~\eqref{eq:controlled-SR} restricts to the
GNH final manifold $C_\infty$ and descends to a controlled presymplectic
system on the quotient $C_\infty/G$.
We now compute the reduced Lie–Poisson system.

Let $z=(f,E,B)\in\g^*_{\mathrm{sdp}}$ denote Eulerian variables.
The momentum map $\mathbf{J}:C_\infty \longrightarrow \g^*$
satisfies $\mathbf{J}(\gamma)=z$ for $\gamma\in C_\infty$. Define the \emph{reduced control vector fields}
\begin{equation}\label{eq:control-generators}
F_a(z)
  := \mathbf{J}_\ast\Big(X_{B_a}\big|_{C_\infty}\Big),
\qquad a=1,\dots,m,
\end{equation}
where each $B_a : W_{\mathrm{red}}\to\mathbb{R}$ is the coupling
functional associated with the $a$-th control channel.  
Because $B$ is $G$-equivariant, each $F_a$ is tangent to coadjoint
orbits.

\begin{lemma}
For each control channel $a$ there exists a map
$\xi_a : \g^*_{\mathrm{sdp}}\to \g_{\mathrm{sdp}}$ such that
\[
F_a(z)=\text{ad}^{*}_{\xi_a(z)}\, z.
\]
\end{lemma}

\begin{proof}
Since~\eqref{eq:affine-cotangent-lift} preserves $G$-orbits, the
restricted control field $X_{B_e^a}|_{C_\infty}$ is tangent to
$G$-orbits in $C_\infty$.  
Pushforward by $\mathbf{J}$ therefore yields a vector tangent to the
coadjoint orbit through $z$, hence of the stated form.
\end{proof}

Thus antenna controls introduce no hidden constraints, and preserve
compatibility with the Lie–Poisson structure.

\begin{theorem}
\label{thm:controlled-LP}
Let $(W,\Omega,H,B)$ be a controlled system with $G$-equivariant
control one-form $B$.  Then:

\begin{enumerate}
\item The controlled flow descends to a well-defined evolution on
$\g^*_{\mathrm{sdp}}\simeq C_\infty/G$.

\item The reduced equations are the \emph{controlled Lie--Poisson system}
\begin{equation}\label{eq:controlled-LP}
\dot{z}
  = \text{ad}^{*}_{\delta H/\delta z} z
    + \sum_{a=1}^m u_a(t)\,\text{ad}^{*}_{\xi_a(z)}z .
\end{equation}

\item The  Casimirs under control are
\[
\mathcal{C}_u
  =\left\{ C \,\middle|\,
      \langle \dd C(z), F_a(z)\rangle = 0
      \ \forall a\right\}.
\]

\item The controlled linearized dynamics at an equilibrium $z_0$ reads
\[
\dot{\delta z}
  = \text{ad}^{*}_{\delta^2 H(z_0)[\delta z]} z_0
    + \sum_{a} u_a(t)\,\text{ad}^{*}_{\xi_a(z_0)}z_0 .
\]
\end{enumerate}
\end{theorem}

\begin{proof}
We prove each statement separately.

\medskip
\noindent\textbf{(1)}
Let $X_{H,u}$ be defined on $W$ by
\[
\iota_{X_{H,u}}\Omega=\dd H + \sum_{a=1}^m u_a(t)\,B_a .
\]
Because $\Omega$ and $H$ are $G$--invariant 
(Lemma~\ref{prop:SR-invariance}) and the control one-form $B$ is 
$G$--equivariant, for every $g\in G$ we have
\[
g^\ast\Omega=\Omega,\qquad g^\ast H=H,\qquad
g^\ast B_a = B_a .
\]
Hence
\[
\iota_{g_\ast X_{H,u}}\Omega
  = \dd H + \sum_{a} u_a B_a
  = \iota_{X_{H,u}}\Omega .
\]
By uniqueness of solutions of the presymplectic equation on $C_\infty$
(GNH theory), this implies
\[
g_\ast X_{H,u} = X_{H,u}\circ g.
\]
Thus the controlled flow preserves $G$--orbits and descends to the 
quotient $C_\infty/G\simeq \mathfrak g_{\mathrm{sdp}}^\ast$.

\medskip
\noindent\textbf{(2)}
On $C_\infty$ the GNH algorithm identifies the dynamics with the 
Marsden--Weinstein Lie--Poisson system on 
$\mathfrak g_{\mathrm{sdp}}^\ast$.  
The projection map 
\(
\mathbf J : C_\infty \to \mathfrak g_{\mathrm{sdp}}^\ast
\)
is the momentum map of the lifted $G$--action.

For the uncontrolled part,
\[
\mathbf J_\ast(X_H)=\operatorname{ad}^{*}_{\delta H/\delta z}\,z .
\]

For each control channel $a$, define $X_{B_a}$ on $C_\infty$ by $\iota_{X_{B_a}}\Omega = B_a$. Since $B_a$ is $G$--equivariant, $X_{B_a}$ is tangent to $G$--orbits in 
$C_\infty$ and hence tangent to coadjoint orbits in 
$\mathfrak g_{\mathrm{sdp}}^\ast$.  
Thus there exists a (non-unique) map
\[
\xi_a(z)\in\mathfrak g_{\mathrm{sdp}}
\qquad\text{such that}\qquad
\mathbf J_\ast(X_{B_a}) = \operatorname{ad}^{*}_{\xi_a(z)} z .
\]

Pushing forward the full controlled vector field,
\[
X_{H,u} = X_H + \sum_{a} u_a(t)\, X_{B_a},
\]
yields
\[
\mathbf J_\ast(X_{H,u})
  = \operatorname{ad}^{*}_{\delta H/\delta z} z
     + \sum_{a} u_a(t)\,\operatorname{ad}^{*}_{\xi_a(z)} z ,
\]
which is exactly the controlled Lie--Poisson equation 
\eqref{eq:controlled-LP}.

\medskip
\noindent\textbf{(3)}
For any functional $C$ on $\mathfrak g_{\mathrm{sdp}}^\ast$, its pullback 
$\mathbf J^\ast C$ is constant along the characteristic distribution $\ker(\Omega_{\mathrm{red}}|_{C_\infty})$. In the controlled system, the characteristic distribution enlarges to
\[
\mathcal K_u
  = \ker(\Omega_{\mathrm{red}}|_{C_\infty})
    \oplus
    \operatorname{span}\{ X_{B_a}|_{C_\infty} \}_{a=1}^m .
\]
Thus $C$ is a (controlled) Casimir iff 
$\dd C(z)$ annihilates the reduced control generators:
\[
\dd C(z)\cdot F_a(z) = 0,
\qquad 
F_a(z)=\mathbf J_\ast(X_{B_a})
      =\operatorname{ad}^{*}_{\xi_a(z)} z .
\]

\medskip
\noindent\textbf{(4)}
Linearizing~\eqref{eq:controlled-LP} around an equilibrium $z_0$ gives
\[
\dot{\delta z}
   = \operatorname{ad}^{*}_{\delta^2 H(z_0)[\delta z]} z_0
      + 
      \operatorname{ad}^{*}_{\delta H(z_0)} \delta z
      +
      \sum_{a} u_a(t)\,\operatorname{ad}^{*}_{\xi_a(z_0)} z_0 .
\]
Since $z_0$ is an equilibrium, 
\(
\operatorname{ad}^{*}_{\delta H(z_0)} z_0 = 0,
\)
and we obtain
\[
\dot{\delta z}
  = \operatorname{ad}^{*}_{\delta^2 H(z_0)[\delta z]} z_0
    + \sum_{a} u_a(t)\,\operatorname{ad}^{*}_{\xi_a(z_0)} z_0 .
\]
\end{proof}


\begin{corollary}
If $\xi_a(z_0)\notin \mathfrak{g}_{z_0}$ (the isotropy algebra of
$z_0$), then the control breaks the residual symmetry and introduces new
tangent directions to $C_\infty/G$.
\end{corollary}

\begin{corollary}
If controls deform Casimirs in a neighborhood of $z_0$, then the
controlled energy–Casimir functional
\[
\mathcal{E}_{C_u}(z)
  := H(z)+C_u(z)
\]
may be tuned so that $z_0$ becomes a strict minimum, yielding formal
stabilization.
\end{corollary}

Theorem~\ref{thm:controlled-LP} therefore provides a rigorous geometric
framework for plasma–antenna interaction, RF heating, and external
stabilization in a manner fully compatible with the Maxwell--Vlasov
Hamiltonian structure.


The generators \(F_a\) encode the intrinsic plasma response to each control channel. 
When their projection lies in the Vlasov (kinetic) sector, the resulting deformation of the distribution function produces resonant phase--space shearing and RF current drive.  
When the projection lies instead in the electromagnetic sector, the induced action corresponds to antenna--driven twisting of magnetic field lines or electrostatic displacement.  
Furthermore, a nonvanishing pairing \(\langle \dd C(z), F_a(z)\rangle \neq 0\) identifies precisely which invariants---such as entropy, magnetic or generalized helicity, and charge-related Casimirs---are broken by the \(a\)-th control channel. Thus, controlled Lie--Poisson reduction naturally incorporates external antenna physics into the SR--GNH geometric framework.

\section{Conclusions and Future Work}

In this paper we have developed a unified geometric formulation of the
Maxwell--Vlasov system based on the infinite-dimensional
Skinner--Rusk (SR) formalism together with the Gotay--Nester--Hinds
(GNH) algorithm for presymplectic constraint analysis.
Starting from the Low Lagrangian on the configuration manifold
$Q=\Diff(T\mathbb{R}^3)\times\mathcal{A}\times\mathcal{V}$, we built the
unified presymplectic phase space $W=TQ\oplus T^*Q$ and showed that the
full Maxwell--Vlasov equations emerge systematically through a finite,
geometrically transparent GNH constraint hierarchy. We then reduced the
SR system by the particle-relabeling symmetry group
$\Diff(T\mathbb{R}^3)$, proving that the reduced presymplectic structure
recovers both the Euler--Poincaré equations of CHHM and the full
Marsden--Weinstein/Morrison--Greene Hamiltonian picture in Eulerian
variables. The resulting framework provides a single geometric origin
for all classical formulations—Lagrangian, Eulerian, Dirac, Lie–Poisson,
and noncanonical Hamiltonian.

We also extended the construction to equilibria that partially
break the relabeling symmetry. By translating the SR system to an
equilibrium and applying a linearized version of the GNH algorithm, we
obtained an effective symplectic phase space free of gauge and relabeling
degeneracies. This construction clarifies the geometric nature of
Goldstone-type modes and provides a presymplectic foundation for
energy–Casimir stability theory around symmetric and partially symmetric
equilibria.

Finally we developed a geometric theory of controlled symmetry
breaking within the SR--GNH framework, incorporating external antenna or
RF drivings as affine Hamiltonian controls compatible with the
presymplectic constraint structure.  After reduction, these controls give
rise to controlled semidirect-product Lie--Poisson equations that modify
Casimir leaves, couple directly to Goldstone-type modes, and permit
stability enhancement via controlled energy--Casimir methods.  This
extends the unified geometric picture to include plasma control and
external actuation, providing a first step toward a rigorous Hamiltonian
framework for RF heating, current drive, and symmetry-breaking
stabilization mechanisms in kinetic plasma models.

\medskip

We conclude by outlining the geometric program suggested by the SR
framework.  All of these directions lie squarely at the intersection of
Hamiltonian control, kinetic theory, and presymplectic reduction.

\paragraph{(1) Controlled Lie--Poisson reduction for plasma--antenna coupling.}

Given an external antenna current $J_{\mathrm{ext}}$, the corresponding
affine cotangent lift defines a controlled action of
$\Diff(T\mathbb{R}^3)$ on $W$.  Reduction by this action yields a
\emph{controlled Lie--Poisson equation}
\[
\dot{z}
  = \operatorname{ad}^{*}_{\delta H/\delta z} z
    \;+\; \sum_{a=1}^m u_a(t)\,F_a(z),
\]
where $F_a$ are the reduced control vector fields.  This provides a
geometric version of radio frequency heating and antenna coupling completely
compatible with the Maxwell--Vlasov Poisson structure.

The central mathematical problem is to characterize controllability and
stabilizability of the Lie--Poisson system relative to
Casimir leaves and symmetry directions.

\paragraph{(2) SR-based Hamiltonian averaging for strong magnetic field.}

Gyrokinetic and drift-kinetic reductions rely on averaging along the
fast Larmor rotation.  In the SR framework this corresponds to
performing presymplectic averaging with respect to the $U(1)$-action
generated by gyroangle rotations.  The gyrocenter Hamiltonian
$H_{\mathrm{gy}}$ becomes
\[
H_{\mathrm{gy}}
  = H_{\mathrm{red}} + \epsilon H^{(1)}_{\mathrm{gy}}
    + \epsilon^2 H^{(2)}_{\mathrm{gy}}
    + \cdots,
\]
with each term obtained via presymplectic projection onto the quotient $C_\infty \big/ U(1)_{\mathrm{gyro}}$.

This yields a rigorous replacement for guiding-center and gyrokinetic
theory within a full Maxwell--Vlasov setting.

\paragraph{(3) Gyrokinetic control via Casimir shaping.}

The Casimirs of the reduced theory determine the foliation of phase
space.  Controlled symmetry breaking modifies these leaves, suggesting a
new method of stabilization: \emph{Casimir shaping}.  The basic idea is
to use controls to deform the Casimir invariants so that an unstable
equilibrium becomes a strict minimum of a controlled energy--Casimir
functional.

Let $C$ be a Casimir and $C_u$ the controlled Casimir after reduction.
Then the controlled second variation is
\[
\delta^2\mathcal{E}_{C_u}(z_0)
  = \delta^2 H(z_0)
    + \delta^2 C_u(z_0).
\]
By choosing $u(t)$ to alter $C_u$ within admissible control directions,
one can produce positivity on the effective linearized space, yielding
formal stabilization.

\medskip

\subsubsection*{Conflict of Interest Statement.}

The author declares that he has no known competing financial interests or personal relationships that could have appeared to influence the work reported in this manuscript.

\subsubsection*{Data Availability Statement.} No datasets were generated or analyzed during the current study. All mathematical derivations are contained within the article.

\end{document}